\documentclass[journal,onecolumn]{IEEEtran}

\usepackage[T1]{fontenc}
\usepackage{amsmath,amssymb,amsthm}
\usepackage{mathrsfs}
\usepackage{enumitem}
\usepackage{xcolor}
\usepackage{physics}
\usepackage{pgfplots}
\pgfplotsset{compat=1.18}
\allowdisplaybreaks
\usepackage{cite}
\usepackage{hyperref}
\hypersetup{
  colorlinks=true,
  urlcolor=blue,
  linkcolor=blue,
  citecolor=blue,
  pdfauthor={Debbie Leung, Benjamin Lovitz, Ion Nechita, Peixue Wu},
  pdftitle={Counterexamples to additivity of minimum output p-Renyi entropy of quantum channels for all p>=0}
}

\newcommand{\mb}{\mathbb}
\newcommand{\mc}{\mathcal}

\newcommand{\diag}{\operatorname{diag}}
\newcommand{\asym}{\mathrm{asym}}
\theoremstyle{plain}
\newtheorem{theorem}{Theorem}[section]
\newtheorem{lemma}[theorem]{Lemma}
\newtheorem{proposition}[theorem]{Proposition}
\newtheorem{corollary}[theorem]{Corollary}
\theoremstyle{definition}
\newtheorem{definition}[theorem]{Definition}

\title{Counterexamples to additivity of minimum output $p$-R\'enyi entropy of quantum channels for all $p\ge 0$}

\author{Debbie~Leung, Benjamin~Lovitz, Ion~Nechita, and~Peixue~Wu%
\thanks{(Corresponding author: Peixue Wu.)}%
\thanks{Debbie Leung is with the Department of Combinatorics and Optimization
and the Institute for Quantum Computing, University of Waterloo, Waterloo,
ON N2L 3G1, Canada, and also with the Perimeter Institute for Theoretical
Physics, Waterloo, ON N2L 2Y5, Canada (e-mail: wcleung@uwaterloo.ca).}%
\thanks{Benjamin Lovitz is with the Department of Computer Science and
Software Engineering, Concordia University, Montreal, QC H3G 1M8, Canada
(e-mail: benjamin.lovitz@concordia.ca).}%
\thanks{Ion Nechita is with the Laboratoire de Physique Th\'eorique, CNRS,
Universit\'e Toulouse~III -- Paul Sabatier, 31062 Toulouse, France
(e-mail: ion.nechita@univ-tlse3.fr).}%
\thanks{Peixue Wu is with the Department of Mathematics and the
Institute for Quantum and Information Sciences, Syracuse University, Syracuse, NY, 13210, USA (e-mail: wpeixue@syr.edu).}%
}

\begin{document}
\maketitle

\begin{abstract}
The additivity of minimum output entropies is a central problem in
quantum information theory. Nonadditivity is known for every R\'enyi
order $p>1$, at the von Neumann point $p=1$, and for sufficiently small
positive $p$, while much of the interval $0<p<1$ has remained open.
In this work, we prove that for every $p>0$ there exist
finite-dimensional quantum channels whose minimum output $p$-R\'enyi
entropies violate additivity. Our proof combines two
constructions: random projection-induced channels yield nonadditivity
for $p>3/4$, and antisymmetric postprocessing extends the violation
to all positive R\'enyi orders. Our estimates also improve the
output-dimension bound obtained by Belinschi, Collins, and Nechita for
additivity violation of minimum output von Neumann entropy.
\end{abstract}


\tableofcontents

\section{Introduction}\label{sec:introduction}

The output entropy under tensor products reveals a
fundamental distinction between classical and quantum channels.
For a quantum channel $\Phi$, its minimum output von Neumann entropy is
\begin{equation}
    H^{\min}(\Phi)
    :=
    \min_{\rho}
    H\!\left(\Phi(\rho)\right),
    \qquad
    H(\sigma):=-\Tr(\sigma\log\sigma),
    \label{eq:min-output-entropy}
\end{equation}
where the minimum is taken over all input density operators.
For two channels, restricting to product inputs gives
\begin{equation}
    H^{\min}(\Phi\otimes\Psi)
    \leq
    H^{\min}(\Phi)+H^{\min}(\Psi).
    \label{eq:min-output-subadditivity}
\end{equation}
Equality always holds for classical stochastic channels. Quantum
channels, however, admit entangled inputs, which may lower the joint
output entropy below the minimum attainable by product inputs. The
minimum-output-entropy additivity conjecture~\cite{AmosovHolevoWerner2000}
asked whether equality in \eqref{eq:min-output-subadditivity}
nevertheless holds for every pair of quantum channels.

More generally, let $S_p$ denote the $p$-R\'enyi entropy for
$0\leq p\leq\infty$, and define
\begin{equation}
    S_p^{\min}(\Phi)
    :=
    \min_{\rho}
    S_p\!\left(\Phi(\rho)\right).
    \label{eq:min-output-renyi}
\end{equation}
The corresponding additivity problem asks for which orders $p$ there
exist channels satisfying
\begin{equation}
    S_p^{\min}(\Phi\otimes\Psi)
    <
    S_p^{\min}(\Phi)+S_p^{\min}(\Psi).
    \label{eq:renyi-additivity-violation}
\end{equation}

At the point $p=1$, we have von Neumann entropy and the additivity
problem is closely connected to classical communication over quantum
channels. The Holevo--Schumacher--Westmoreland
theorem~\cite{SchumacherWestmoreland1997,Holevo1998} expresses the
unassisted classical capacity as the regularized Holevo quantity:
\begin{equation}
    C(\Phi)
    =
    \lim_{n\to\infty}
    \frac{1}{n}\chi\!\left(\Phi^{\otimes n}\right),
    \label{eq:regularized-holevo-capacity}
\end{equation}
where $\chi$ denotes the (maximal) Holevo quantity. Shor proved that, as universal statements over all channels and bipartite states, the universal additivity conjectures for the
minimum output von Neumann entropy, the Holevo capacity, and
the entanglement of formation are equivalent~\cite{Shor2004}.
These equivalences connect the entropy additivity problem to both
channel coding and entanglement theory.

The first counterexamples were found away from $p=1$. Werner and
Holevo constructed explicit counterexamples for
$p>4.79$~\cite{WernerHolevo2002}, and Hayden and Winter subsequently
proved nonadditivity for every $p>1$ using random
channels~\cite{Hayden_2008}. Hastings then disproved additivity at
$p=1$ through a random finite-dimensional
construction~\cite{Hastings_2009}. Subsequent work using concentration
of measure and asymptotic geometric analysis clarified the mechanisms
behind these random constructions and related them to almost Euclidean
sections of high-dimensional convex
bodies~\cite{BrandaoHorodecki2010,FukudaKingMoser2010,
AubrunSzarekWerner2010,AubrunSzarekWerner2011}. Free probability provided
sharper asymptotic entropy estimates and quantitative bounds on the
output dimensions needed to witness
nonadditivity~\cite{Belinschi_2012,Belinschi_2016,CollinsFukudaNechita2015}.

Constructive counterexamples
for every $p>2$ were obtained in
Refs.~\cite{Grudka_2010,Szczygielski_2024}, and a recent result extends
such constructions to every $p>1$~\cite{Derksen_2026}. At the opposite
end of the R\'enyi scale, Cubitt, Harrow, Leung, Montanaro, and Winter
proved nonadditivity at $p=0$ and, consequently, for all sufficiently
small positive $p$~\cite{Cubitt_2008}. They also constructed an explicit
pair of channels with input dimension $4$ and output dimension $3$
whose minimum output rank is nonmultiplicative. Numerical evidence
suggested that this pair continues to violate additivity up to
approximately $p=0.11$, but no rigorous endpoint was established.

The intermediate regime \(0<p<1\) has therefore remained substantially less
understood.  A fixed finite-dimensional violation at \(p=1\) persists,
by continuity, on an interval \((1-\varepsilon,1]\) depending on the
chosen pair of channels.  This observation gives neither an explicit
value of \(\varepsilon\) nor a random-channel mechanism controlled
through the von Neumann point.  Our main result closes the full interval
\(0<p<1\).

\begin{theorem}\label{thm:main}
For any \(p>0\), there exist finite-dimensional quantum channels
\(\Phi,\Psi\) such that
\begin{equation}
    S_p^{\min}
    \!\left(\Phi\otimes\Psi\right)
    <
    S_p^{\min}(\Phi)+S_p^{\min}(\Psi).
    \label{eq:main-violation}
\end{equation}
\end{theorem}
The projection-induced construction first yields the range \(p>3/4\).
The remaining interval \(0<p\le3/4\) is then covered by composing with a
deterministic postprocessing channel whose outputs are supported on
antisymmetric subspaces.  With a suitable choice of the antisymmetric
parameters, nonadditivity holds for all \(p>0\).  Moreover, a
high-precision evaluation of the resulting asymptotic criterion shows
that, for the von Neumann entropy, a violation already occurs at output
dimension \(k=182\), one below the threshold \(k=183\) established
in~\cite{Belinschi_2016} for the Haar random subspace construction.
 
\subsection{Overview of the proof}
Our counterexamples are derived from random channels from $A \cong \mb C^n$ to $B \cong \mb C^k$ obtained from
Haar-random projections.  Fix \(k\geq2\) and \(t\in(1/k^2,1)\), and let
\(P_n\) be a Haar-random projection of rank \(d_n\) on
\(\mb C^n\otimes\mb C^k\), where \(d_n/(kn)\to t\) as \(n\to\infty\).
Writing \(P_{A,n}=\Tr_BP_n\), we locally normalize \(P_n\) to obtain the
Choi operator
\begin{equation}
    J_n
    =
    (P_{A,n}^{-1/2}\otimes I_k)
    P_n
    (P_{A,n}^{-1/2}\otimes I_k).
    \label{eq:intro-normalized-choi}
\end{equation}
The condition \(t>1/k^2\) ensures that \(P_{A,n}\) is almost surely
invertible for all sufficiently large \(n\).  Moreover,
\(J_n\geq0\) and \(\Tr_BJ_n=I_n\), so \(J_n\) is the Choi matrix of the
quantum channel
\[
    \Phi_n(X)
    =
    \Tr_A\!\left[J_n(X^{T}\otimes I_k)\right],
\]
which is the projection-induced channel \(\Phi_{P_n}\) of
Section~\ref{sec:prelim}.
 
\subsubsection{Exact asymptotic evaluation of one-copy minimum output entropy}
Our first step is to determine the deterministic large-\(n\) limit of
the output set, which is a convex body:
\[
    \mathscr C_n:=\Phi_n(\mc D_n)\subseteq\mc D_k,
\]
where \(\mc D_m\) denotes the set of density operators on \(\mb C^m\).
Using strong asymptotic freeness and the free-probabilistic techniques
of~\cite{CollinsMale2014,CollinsFukudaNechita2015,
ArizmendiNechitaVargas2016,Nechita2018,CollinsNechita2016Review}, we
prove in Theorem~\ref{thm:output-space-limit} that, almost surely,
\[
    \mathscr C_n\rightarrow\mathscr K_{k,t}
\]
in Hausdorff distance, where
\[
\mathscr K_{k,t}=\left\{\frac{X}{\Tr X}:0\le X\le I_k,\ \Tr c_t(X)\le\frac{1}{k}\right\},\quad c_t(u):=\left(\sqrt{t(1-u)}-\sqrt{(1-t)u}\right)^2.
\]
Consequently, for every fixed \(p>0\),
\[
    \lim_{n\to\infty}S_p^{\min}(\Phi_n)
    =
    S_p^{\min}(\mathscr K_{k,t}),
    \quad\text{almost surely},
\]
where
\(S_p^{\min}(\mathscr K_{k,t}):=\min_{\rho\in\mathscr K_{k,t}}S_p(\rho)\).
A subsequent large-\(k\) analysis, with \(p>0\) and \(t\) fixed, gives
\begin{equation}
    S_p^{\min}(\mathscr K_{k,t})
    =
    \log k-\frac{2p(1-t)}{tk^2}
    +o(k^{-2}).
    \label{eq:intro-single-asymptotic}
\end{equation}
 
\subsubsection{Exact asymptotic evaluation of two-copy Bell-output entropy}
We pair \(\Phi_n\) with its complex-conjugate channel
\(\overline\Phi_n\) and evaluate the product channel at a maximally
entangled input, the Bell state.  Let
\[
    |\psi_m^+\rangle
    =
    \frac{1}{\sqrt m}\sum_{i=1}^m|i\rangle\otimes|i\rangle,
    \qquad
    \psi_m^+=|\psi_m^+\rangle\langle\psi_m^+|,\qquad\forall m\ge1.
\]
In Theorem~\ref{thm:bell-output-limit}, we show that the Bell output
\((\Phi_n\otimes\overline\Phi_n)(\psi_n^+)\) converges almost surely to
the isotropic state
\[
    Z_{k,t}
    =
    r_{k,t}\psi_k^+
    +(1-r_{k,t})\frac{I_{k^2}}{k^2},
    \qquad
    r_{k,t}
    =
    \frac{k^2(1-t)}
    {(k^2-1)^2t+1-t}.
\]
Thus \(Z_{k,t}\) has eigenvalues
\[
    \alpha_{k,t}
    =
    r_{k,t}+\frac{1-r_{k,t}}{k^2},
    \qquad
    \beta_{k,t}
    =
    \frac{1-r_{k,t}}{k^2},
\]
with respective multiplicities \(1\) and \(k^2-1\), and its entropy
satisfies
\begin{equation}
    \lim_{n\to\infty}
    S_p\!\left(
        (\Phi_n\otimes\overline{\Phi}_n)(\psi_n^+)
    \right)
    =
    2\log k-\frac{A_p(t)}{k^2}+o(k^{-2}),
    \label{eq:intro-joint-asymptotic}
\end{equation}
where
\begin{equation}
    A_p(t)
    =\frac{t^{-p}-1-p(t^{-1}-1)}{p-1},
    \label{eq:intro-Apt}
\end{equation}
and \(A_1(t)\) is defined by continuous extension.  Since
\(S_p^{\min}(\overline\Phi_n)=S_p^{\min}(\Phi_n)\), comparing
\eqref{eq:intro-single-asymptotic} and \eqref{eq:intro-joint-asymptotic}
reduces the Bell-state test to
\begin{equation}
    A_p(t)>\frac{4p(1-t)}{t}.
    \label{eq:intro-threshold-condition}
\end{equation}
An asymptotic analysis shows that \eqref{eq:intro-threshold-condition}
holds for all sufficiently small \(t\) whenever \(p>3/4\), see Theorem~\ref{thm:entropy-asymptotics-antisymmetric} when $r=1$; this includes
\(p=1\) thus the von Neumann point is not singular for the
present method and does not require a separate Hastings-type
argument~\cite{Hastings_2009,aubrun2017alice}.
 
\subsubsection{Extending additivity violation to all \(p>0\) via deterministic antisymmetric postprocessing}

To reach the remaining orders, we postprocess \(\Phi_n\) by a
deterministic channel whose output is supported on the antisymmetric
subspace of \((\mb C^k)^{\otimes r}\).  Let \(\mathfrak S_r\) be the
symmetric group on \(r\) letters, and for \(\pi\in\mathfrak S_r\) let
\(P_\pi\) be the unitary that permutes the \(r\) tensor factors
according to \(\pi\):
\[
P_\pi|j_1j_2\cdots j_r\rangle=|j_{\pi^{-1}(1)}j_{\pi^{-1}(2)}\cdots j_{\pi^{-1}(r)}\rangle,
\]
where \(j_1,\ldots,j_r\in[k]\).  Define the projection
\[
    \Pi_{k,r}^{\asym}:=
    \frac1{r!}\sum_{\pi\in\mathfrak S_r}
        \operatorname{sgn}(\pi)P_\pi,
    \quad
    \mc A_{k,r}^{\asym}:=
    \Pi_{k,r}^{\asym}\bigl(\mb C^k\bigr)^{\otimes r}.
\]
Thus \(\mc A_{k,r}^{\asym}\) is the antisymmetric subspace, of dimension
\(\binom kr\).  The postprocessing channel is
\begin{equation}\label{def:channel-postprocessing}
    \mc E_{k,r}(X):=\frac{k}{\binom{k}{r}}\Pi_{k,r}^{\asym}\left(X\otimes I_k^{\otimes(r-1)}\right)\Pi_{k,r}^{\asym}.
\end{equation}
Put
\[
    \Theta_n:=\mc E_{k,r}\circ\Phi_n.
\]
Let \(\gamma:=(1-t)/t\).  For fixed \(r\) and \(t\), we show that, as
\(k\to\infty\),
\begin{align}
& \lim_{n\to\infty}S_p^{\min}(\Theta_n)
    =\log\binom{k}{r}-\frac{2p\gamma}{rk^2}+o(k^{-2}),
        \label{eq:intro-exterior-single}\\
& \lim_{n\to\infty}S_p\bigl((\Theta_n\otimes\overline\Theta_n)(\psi_n^+)\bigr)
    =2\log\binom{k}{r}
        -\frac{B_{p,r}(\gamma)}{k^2}+o(k^{-2}),
        \label{eq:intro-exterior-joint}
\end{align}
where
\[
    B_{p,r}(\gamma)
    :=
    \frac{p\gamma-r^2\bigl((1+\gamma/r^2)^p-1\bigr)}{1-p},
\]
and \(B_{1,r}(\gamma)\) is defined by continuous extension.  For
\(r=1\), the channel \(\mc E_{k,1}\) is the identity and
\(B_{p,1}(\gamma)=A_p(t)\), so \eqref{eq:intro-exterior-single} and
\eqref{eq:intro-exterior-joint} reduce to
\eqref{eq:intro-single-asymptotic} and
\eqref{eq:intro-joint-asymptotic}.  Comparing
\eqref{eq:intro-exterior-single} and \eqref{eq:intro-exterior-joint}
reduces the Bell-state test to
\begin{equation}\label{eq:intro-threshold-condition-post}
    B_{p,r}(\gamma)>\frac{4p\gamma}{r}.
\end{equation}
As $r$ increases, the above inequality holds, but in particular, a direct computation shows that the single choice
\[
    r=5,\qquad\gamma=375
\]
satisfies \eqref{eq:intro-threshold-condition-post} for every
\(p\in(0,\infty)\), which proves Theorem~\ref{thm:main} for all
\(p>0\).
 
\subsection{Comparison with previous work}
The use of a maximally entangled input for a channel and its complex
conjugate was introduced in~\cite{WernerHolevo2002} and subsequently
used in~\cite{Hayden_2008,Hastings_2009,aubrun2017alice}.  It was later
analyzed systematically by Collins and
Nechita~\cite{Collins_2010,CollinsNechita2011}, who established the
Bell-state phenomenon for random channels and computed the limiting
outputs of conjugate pairs of random channels.  For the
Haar--Stinespring ensemble, equivalently the Haar-random-subspace model,
Belinschi, Collins, and Nechita subsequently identified the
deterministic limit of the set of output spectra~\cite{Belinschi_2012},
which differs from our \(\mathscr K_{k,t}\), and optimized the relevant
norms and entropies over this set~\cite{Belinschi_2016}.  In the von
Neumann case, they proved that the conjugate pair yields a violation at
output dimension \(k=183\), whereas, almost surely, the same comparison
cannot yield a violation for \(k\le182\) in that ensemble.  They also
obtained, in a suitable asymptotic regime, entropy gaps arbitrarily
close to \(\log2\).  Convergence of channel output sets was later placed
in a general compact-convex framework by Collins, Fukuda, and
Nechita~\cite{CollinsFukudaNechita2015}.
 
The ensemble considered here is different: we begin with a Haar-random
projection and use its local normalization as a Choi matrix.  This
normalization changes both the limiting one-channel output body and the
limiting isotropic Bell-output spectrum.  For the present ensemble, a
high-precision evaluation of the resulting asymptotic Bell-state
criterion shows a von Neumann-entropy violation already at \(k=182\),
one dimension below the threshold in~\cite{Belinschi_2016}.  For broader
accounts of the random-matrix and asymptotic-convex-geometric methods
surrounding these results, see
Refs.~\cite{CollinsNechita2016Review,aubrun2017alice}.
 
The channel \(\mc E_{k,r}\) defined in \eqref{def:channel-postprocessing}
has previously appeared as the block \(\mc G_{k,r}\) of the Grassmann
channel introduced in~\cite{BradlerJochymOConnorJauregui2011}; see in
particular Definition~5 therein.  Equivalently, it is the \(1\to r\)
antisymmetric cloning channel obtained by specializing the
subspace-cloning construction of \cite[Eq.~(15)]{LemmWilde2017} to
\(Y_1=\mb C^k\) and \(X_r=\mc A_{k,r}^{\asym}\).
 
For completeness, we note that a 2010 preprint by Yu and Ying announced
nonadditivity for \(p\in(0,p_0)\cup(1-p_0,1)\) with \(p_0\simeq0.2855\).
The preprint was subsequently withdrawn because of what the authors
described as a crucial error, so the announced ranges are not regarded
as established~\cite{YuYingWithdrawn}.
 
The remainder of the paper is organized as follows.
Section~\ref{sec:prelim} collects preliminaries, including
projection-induced channels, Haar-random projections, and the random
compression formula.  Section~\ref{sec:output-space} determines the
limiting one-channel output body and hence the one-copy minimum output
entropy.  Section~\ref{sec:bell-output} computes the conjugate-pair
Bell output, and Section~\ref{sec:antisymmetric-postprocessing} analyzes
how these entropies change under antisymmetric postprocessing.  Finally,
Section~\ref{sec:main-proof} proves Theorem~\ref{thm:main}.  The
appendices collect the basics of free probability, the asymptotic
entropy calculations, and the numerical certificate of the improved output dimension threshold.

\section{Preliminaries}\label{sec:prelim}
 
We fix Hilbert spaces \(\mc H_A=\mb C^n\) and \(\mc H_B=\mb C^k\), with
standard bases \(\{|a\rangle_A\}_{a=1}^n\) and \(\{|i\rangle_B\}_{i=1}^k\),
and write \(|ai\rangle=|a\rangle_A\otimes|i\rangle_B\).  Transposition and
entrywise complex conjugation are always taken in these bases.  We write
\(\mc L(A,B)\) for the linear maps \(\mc H_A\to\mc H_B\),
\(\mc B(A)=\mc L(A,A)\), and
\[
   \mc D(A)=\mc D_n=\{\rho\in\mc B(A):\rho\geq0,\ \Tr\rho=1\},
\]
and similarly \(\mc B(B)\), \(\mc D(B)=\mc D_k\), and
\(\mc B(AB)=\mc B(\mc H_A\otimes\mc H_B)\).  We identify \(\mc B(A)\) and
\(\mc B(B)\) with the matrix algebras \(\mb M_n\) and \(\mb M_k\), write
\(\mb M_m^{\rm sa}\) for the Hermitian matrices in \(\mb M_m\), and put
\(I_A=I_n\) and \(I_B=I_k\).  The trace norm and the operator norm are
denoted by \(\|\cdot\|_1\) and \(\|\cdot\|_\infty\).
 
We distinguish the unnormalized Bell vector from the normalized Bell
state:
\[
 |\Omega_A\rangle=\sum_{a=1}^n|a,a\rangle,
 \qquad
 |\psi_A^+\rangle=\frac1{\sqrt n}|\Omega_A\rangle,
 \qquad
 \psi_A^+=|\psi_A^+\rangle\langle\psi_A^+|.
\]
The (unnormalized) Choi matrix of a linear map
\(\Phi:\mc B(A)\to\mc B(B)\) is
\begin{align*}
 J_\Phi
 &=(\mathrm{id}_A\otimes\Phi)
   (|\Omega_A\rangle\langle\Omega_A|) =\sum_{a,b=1}^n|a\rangle\langle b|
   \otimes\Phi(|a\rangle\langle b|)
 =n(\mathrm{id}_A\otimes\Phi)(\psi_A^+).
\end{align*}
Conversely,
\[
 \Phi(X)=\Tr_A[J_\Phi(X^T\otimes I_B)].
\]
The map \(\Phi\) is a quantum channel if and only if \(J_\Phi\geq0\) and
\(\Tr_BJ_\Phi=I_A\).  The conjugate channel is
\(\overline\Phi(X)=\overline{\Phi(\overline X)}\), so that
\(J_{\overline\Phi}=\overline{J_\Phi}\).
 
For \(0<p<\infty\), \(p\neq1\), the R\'enyi entropy is
\[
 S_p(\rho)=\frac{1}{1-p}\log\Tr(\rho^p).
\]
We use the continuous and endpoint conventions
\[
 S_0(\rho)=\log\rank\rho,
 \qquad
 S_1(\rho)=-\Tr(\rho\log\rho),
 \qquad
 S_\infty(\rho)=-\log\lambda_{\max}(\rho),
\]
and define
\[
 S_p^{\min}(\Phi)
 =\min_{\rho\in\mc D(A)}S_p(\Phi(\rho)).
\]
All logarithms are natural.
 
Let \(R_{AB}\in\mc B(AB)\) be a positive semidefinite operator, and
assume that
\[
    R_A:=\Tr_BR_{AB}
\]
is strictly positive on \(\mc H_A\).  Define
\[
    J_R
    :=
    (R_A^{-1/2}\otimes I_B)
    R_{AB}
    (R_A^{-1/2}\otimes I_B).
\]
Then \(J_R\geq0\) and \(\Tr_BJ_R=R_A^{-1/2}R_AR_A^{-1/2}=I_A\), so
\(J_R\) is the Choi matrix of the quantum channel
\begin{equation}
    \Phi_R:\mc B(A)\to\mc B(B),\quad
    \Phi_R(X):=\Tr_A\left[J_R(X^T\otimes I_B)\right].
    \label{eq:def-Phi-R}
\end{equation}
We call \(R_{AB}\) a generalized Choi operator for \(\Phi_R\).  The
construction is invariant under positive rescaling: \(J_{cR}=J_R\) for
\(c>0\).
 
\begin{definition}[Projection-induced channels]\label{def:projection-induced}
We call \(\Phi:\mc B(A)\to\mc B(B)\) a projection-induced channel if
\(\Phi=\Phi_P\), where \(P=P_{AB}\) is a projection on
\(\mc H_A\otimes\mc H_B\) such that
\[
    P_A:=\Tr_BP_{AB}>0.
\]
\end{definition}
 
We will study Haar-random projections on
\(\mc H_A\otimes\mc H_B=\mb C^n\otimes\mb C^k\).  We say that
\(P\in\mc B(AB)\) is a Haar-random projection of rank \(d\geq1\) if
\[
    P=UP_0U^*,
\]
where \(U\) is a Haar-random unitary on \(\mc H_A\otimes\mc H_B\) and
\(P_0\) is a fixed projection of rank \(d\).  A projection \(P_{AB}\)
induces a channel via Definition~\ref{def:projection-induced} only if
\(P_A>0\).  For Haar-random projections, the following lemma shows when
this holds almost surely.
 
\begin{lemma}[Full local support of a Haar-random subspace]
\label{lem:haar-projection-partial-trace}
Let \(n,k,d\in\mb N\) with \(1\leq d\leq nk\), and let \(P=P_{AB}\) be a
Haar-random projection of rank \(d\) on \(\mb C^n\otimes\mb C^k\).  Then
\[
    \rank(\Tr_BP)=\min\{n,kd\}
    \qquad\text{almost surely}.
\]
Consequently,
\[
    \Tr_BP>0\quad\text{almost surely}
    \qquad\Longleftrightarrow\qquad
    n\leq kd.
\]
\end{lemma}
 
\begin{proof}
By the Gaussian realization of Haar measure on the Grassmannian, see for
example \cite[Exercise~B.14]{aubrun2017alice} and
\cite[Section~5]{Mezzadri2007}, the support of \(\Tr_BP\) agrees with
that of the reduced state of a Haar-random vector in
\(\mb C^n\otimes\mb C^{kd}\).  The generic-rank result for induced random
states therefore gives
\[
    \rank(\Tr_BP)=\min\{n,kd\}
\]
almost surely; see
\cite[Sec.~III.A, paragraph following Eq.~(3.1)]{ZyczkowskiSommers2001}.
\end{proof}

The following asymptotic result for Haar random projections is crucial.
It is derived using the techniques developed in
\cite{CollinsMale2014,ArizmendiNechitaVargas2016,Nechita2018}.  Readers
who are unfamiliar with free probability may wish to start with
Appendix~\ref{app:free-probability}.
 
\begin{lemma}[Random compression formula]
\label{lem:random-matrix-input}
Let \(n,k,d\in\mb N\) with \(1\leq d\leq nk\), and let
    \(P_n\) be a Haar-random projection of rank \(d_n\) on
    \(\mb C^n\otimes\mb C^k\), where
$d_n/(kn)\to t\in(0,1)$.  In the standard tensor product basis, write
$P_n=\sum_{i,j=1}^kP_{ij}^{(n)}\otimes|i\rangle\langle j|$ with
$P_{ij}^{(n)}\in\mb M_n$.  Then, on an event of probability one,
simultaneously for every $a=(a_1,\ldots,a_k)\in\mb R^k$,
\begin{equation}
    \lim_{n\to\infty}
    \lambda_{\max}\left(\sum_{i=1}^ka_iP_{ii}^{(n)}\right)
    =
    \max_{u\in\mathscr D_{k,t}}
    \sum_{i=1}^ka_iu_i,
    \label{eq:random-compression}
\end{equation}
where
\begin{equation}
    \mathscr D_{k,t}
    :=
    \left\{
      u\in[0,1]^k:
      \sum_{i=1}^k
      \left(\sqrt{t(1-u_i)}-\sqrt{(1-t)u_i}\right)^2
      \le\frac1k
    \right\}.
    \label{eq:set-D}
\end{equation}
\end{lemma}
 
\begin{proof}
For $a\in\mb R^k$, put $A:=\operatorname{diag}(a_1,\ldots,a_k)$ and
$S_n(a):=\sum_{i=1}^ka_iP_{ii}^{(n)}$, and let $M_a$ denote the
right-hand side of \eqref{eq:random-compression}.  We first prove
\eqref{eq:random-compression} almost surely for each fixed $a$, and then
remove the dependence of the event on $a$.
 
Fix $a\in\mb R^k$, and consider the Hermiticity-preserving map
$$\varphi_a\colon\mb M_k\to\mb M_k,\quad\varphi_a(X):=\Tr(AX)I_k.$$  The
Choi matrix of $\varphi_a$ is $I_k\otimes A$, whose spectral projections
have partial traces proportional to $I_k$.  Thus $\varphi_a$ satisfies
the unitarity condition of \cite[Definition~5.1]{Nechita2018}.  Let
$b_t:=(1-t)\delta_0+t\delta_1$ be the Bernoulli measure with parameter
$t$.  Since $P_n$ converges strongly to $b_t$ (see
Appendix~\ref{app:sub-basics}), the strong block-modification theorem
\cite[Theorem~5.2]{Nechita2018}, together with
\cite[Theorem~5.1]{ArizmendiNechitaVargas2016}, shows that
$(\operatorname{id}_n\otimes\varphi_a)(P_n)=S_n(a)\otimes I_k$
converges strongly, almost surely, to
$\mu_{a,t}:=\boxplus_{i=1}^k(D_{a_i/k}b_t)^{\boxplus k}$.  Every
eigenvalue of $S_n(a)$ occurs exactly $k$ times in $S_n(a)\otimes I_k$,
so strong convergence yields
$\lambda_{\max}(S_n(a))\to\max\operatorname{supp}(\mu_{a,t})$ almost
surely.  By Lemma~\ref{lem:bernoulli-edge-duality},
\[
\max\operatorname{supp}(\mu_{a,t})=  \max_{u\in\mathscr D_{k,t}}
    \sum_{i=1}^ka_iu_i =:M_a.
\]
It remains to choose a single event for all $a$.  Put
$\|a\|_1:=\sum_{i=1}^k|a_i|$.  Since
$P_{ii}^{(n)}=(I_n\otimes\langle i|)P_n(I_n\otimes|i\rangle)$ and
$0\le P_n\le I$, we have $0\le P_{ii}^{(n)}\le I_n$.  Hence, for all
$a,b\in\mb R^k$, we have
$$|\lambda_{\max}(S_n(a))-\lambda_{\max}(S_n(b))|
\le\|S_n(a)-S_n(b)\|\le\|a-b\|_1.$$  Since
$\mathscr D_{k,t}\subseteq[0,1]^k$, we also have
$|M_a-M_b|\le\|a-b\|_1$.  Let $\Omega_0$ be the intersection of the
probability-one events obtained above for $a\in\mb Q^k$; since $\mb Q^k$
is countable, $\Omega_0$ has probability one.  On $\Omega_0$, the
triangle inequality gives, for every $a\in\mb R^k$ and $b\in\mb Q^k$,
$\limsup_{n\to\infty}|\lambda_{\max}(S_n(a))-M_a|\le2\|a-b\|_1$.  Since
$\mb Q^k$ is dense in $\mb R^k$, the left-hand side vanishes.  Thus
\eqref{eq:random-compression} holds on $\Omega_0$ simultaneously for
every $a\in\mb R^k$.
\end{proof}
\section{Output states of Haar-random projection-induced channels}
\label{sec:output-space}
 
Fix an integer $k\geq2$.  For each $n$, let $P_n=P_{AB,n}$ be an
independent Haar-distributed projection of rank $d_n$ on
$\mb C^n\otimes\mb C^k$, where
\begin{equation}
    \frac{d_n}{kn}\longrightarrow t\in\left(\frac1{k^2},1\right),
    \label{eq:projection-density}
\end{equation}
and set $P_{A,n}:=\Tr_BP_n$. By Lemma~\ref{lem:haar-projection-partial-trace}, almost surely, $P_{A,n}$ is invertible for all sufficiently large $n$, and we consider
the channel
\begin{equation}
    \Phi_n(\rho)
    :=
    \Tr_A\left[
      P_n\left(P_{A,n}^{-1/2}\rho^TP_{A,n}^{-1/2}\otimes I_k\right)
    \right].
    \label{eq:def-channel}
\end{equation}
The goal of this section is to evaluate
$\lim_{n\to\infty}S_p^{\min}(\Phi_n)$ exactly. Recall $\mathscr D_{k,t}$ in \eqref{eq:set-D}, define
\begin{align}
    \Lambda_{k,t}
    &:=
    \left\{\frac{u}{\sum_iu_i}:u\in\mathscr D_{k,t}\right\},
    \label{eq:def-Lambda-kt}\\
    \mathscr K_{k,t}
    &:=
    \left\{U\diag(\lambda)U^*:\lambda\in\Lambda_{k,t},\ U\in\mathcal U(k)\right\}.
    \label{eq:def-K}
\end{align}
Equivalently, with $c_t(u):=\left(\sqrt{t(1-u)}-\sqrt{(1-t)u}\right)^2$ applied by functional calculus,
\begin{equation}
    \mathscr K_{k,t}
    =
    \left\{
      \frac{X}{\Tr X}:
      0\le X\le I_k,\ \Tr c_t(X)\le\frac1k
    \right\}.
    \label{eq:operator-form-K}
\end{equation}
Since $\sum_iu_i\ge m_{k,t}>0$ on $\mathscr D_{k,t}$, the map
$u\mapsto u/\sum_iu_i$ is continuous on $\mathscr D_{k,t}$, so
$\Lambda_{k,t}$ and $\mathscr K_{k,t}$ are compact.  Moreover,
$\Lambda_{k,t}$ is the intersection of the probability simplex with the
cone $\{su:s>0,\ u\in\mathscr D_{k,t}\}$, which is convex because
$\mathscr D_{k,t}$ is.  Similarly, since $X\mapsto\Tr c_t(X)$ is convex,
\eqref{eq:operator-form-K} exhibits $\mathscr K_{k,t}$ as the
intersection of $\mc D_k$ with a convex cone.  Hence $\Lambda_{k,t}$ and
$\mathscr K_{k,t}$ are compact and convex.  For compact sets
$C,K\subseteq\mb M_k^{\rm sa}$, we use the Hausdorff distance induced by
the trace norm,
\begin{equation}
    d_H^{(1)}(C,K)
    :=
    \max\left\{
      \sup_{\rho\in C}\inf_{\sigma\in K}\|\rho-\sigma\|_1,\
      \sup_{\sigma\in K}\inf_{\rho\in C}\|\rho-\sigma\|_1
    \right\}.
    \label{eq:trace-Hausdorff}
\end{equation}
 
\begin{theorem}\label{thm:output-space-limit}
Almost surely,
$d_H^{(1)}(\Phi_n(\mc D_n),\mathscr K_{k,t})\to0$.
\end{theorem}
 
\begin{corollary}\label{cor:single-asymptotic}
For every fixed $p\in(0,\infty)$, almost surely,
\begin{equation}
    \lim_{n\to\infty}S_p^{\min}(\Phi_n)
    =
    \min_{\sigma\in\mathscr K_{k,t}}S_p(\sigma).
    \label{eq:single-fixed-k-limit}
\end{equation}
Moreover, for fixed $t\in(0,1)$ and $p\in(0,\infty)$, as $k\to\infty$,
\begin{equation}
    \min_{\sigma\in\mathscr K_{k,t}}S_p(\sigma)
    =
    \log k-\frac{2p(1-t)}{tk^2}+o(k^{-2}).
    \label{eq:single-large-k}
\end{equation}
\end{corollary}
 
\begin{proof}
We have $S_p^{\min}(\Phi_n)=\min_{\sigma\in\Phi_n(\mc D_n)}S_p(\sigma)$,
and $S_p$ is uniformly continuous on $\mc D_k$ for fixed $k$ and
$p\in(0,\infty)$.  Hence the first assertion follows from
Theorem~\ref{thm:output-space-limit}.  The large-$k$ asymptotics are
derived in Appendix~\ref{app:entropy-asymptotics}.
\end{proof}
 
We now prove Theorem~\ref{thm:output-space-limit}.  The proof is based
on support functions.  For a compact convex set
$\mathscr K\subseteq\mb M_k^{\rm sa}$, the support function
$h_{\mathscr K}$ is defined by
$h_{\mathscr K}(H):=\max_{\sigma\in\mathscr K}\Tr(H\sigma)$ for
$H\in\mb M_k^{\rm sa}$.  Let
\begin{equation}\label{eq:unit-ball-operator-norm}
    \mathcal B:=\{H\in\mb M_k^{\rm sa}:\|H\|_\infty\le1\}
\end{equation}
be the operator-norm unit ball.
 
\begin{lemma}\label{lem:hausdorff-support}
For compact convex sets $C,K\subseteq\mb M_k^{\rm sa}$,
\begin{equation}
    d_H^{(1)}(C,K)
    =
    \sup_{H\in\mathcal B}
    \bigl|h_C(H)-h_K(H)\bigr|.
    \label{eq:hausdorff-support-duality}
\end{equation}
\end{lemma}
 
\begin{proof}
By trace--operator norm duality,
$\|\sigma-\tau\|_1=\max_{H\in\mathcal B}\Tr(H(\sigma-\tau))$ for all
$\sigma,\tau\in\mb M_k^{\rm sa}$.  Since $K$
and $\mathcal B$ are compact and convex, the minimax theorem yields
$\inf_{\tau\in K}\|\sigma-\tau\|_1
=\max_{H\in\mathcal B}\bigl(\Tr(H\sigma)-h_K(H)\bigr)$.  Taking the
supremum over $\sigma\in C$, we obtain
$\sup_{\sigma\in C}\inf_{\tau\in K}\|\sigma-\tau\|_1
=\sup_{H\in\mathcal B}\bigl(h_C(H)-h_K(H)\bigr)$.  Exchanging the roles
of $C$ and $K$ and combining the two identities proves
\eqref{eq:hausdorff-support-duality}.
\end{proof}
 
Write $\mathscr C_n:=\Phi_n(\mc D_n)$, which is a compact convex subset
of $\mc D_k$, and $h_n:=h_{\mathscr C_n}$.  For $H\in\mb M_k$, define
\begin{equation}
    T_n(H):=\Tr_B\bigl(P_n(I_n\otimes H)\bigr)\in\mb M_n,
    \label{eq:def-T-map}
\end{equation}
so that $T_n(I_k)=P_{A,n}$.
 
\begin{lemma}\label{lem:support-channel}
Let $n$ be such that $P_{A,n}$ is invertible.  For every
$H\in\mb M_k^{\rm sa}$,
\begin{equation}
\begin{split}
    h_n(H)
    &=
    \lambda_{\max}\bigl(P_{A,n}^{-1/2}T_n(H)P_{A,n}^{-1/2}\bigr)\\
    &=
    \inf\bigl\{\lambda\in\mb R:
      \lambda_{\max}\bigl(T_n(H-\lambda I_k)\bigr)\le0\bigr\}.
\end{split}
    \label{eq:support-channel}
\end{equation}
\end{lemma}
 
\begin{proof}
Let $\rho\in\mc D_n$, and put
$M:=P_{A,n}^{-1/2}\rho^TP_{A,n}^{-1/2}$.  By the cyclicity of the trace
and of the partial trace over $B$, we have
$$\Tr(H\Phi_n(\rho))=\Tr\bigl((M\otimes H)P_n\bigr)=\Tr(MT_n(H)),$$ which
equals $\Tr\bigl(P_{A,n}^{-1/2}T_n(H)P_{A,n}^{-1/2}\rho^T\bigr).$  Since
$\rho\mapsto\rho^T$ is a bijection of $\mc D_n$, the first equality
follows.  Moreover, $h_n(H)\le\lambda$ if and only if
$P_{A,n}^{-1/2}T_n(H)P_{A,n}^{-1/2}\le\lambda I_n$, that is,
$T_n(H)\le\lambda P_{A,n}=\lambda T_n(I_k)$, or equivalently
$T_n(H-\lambda I_k)\le0$.  This gives the second equality.
\end{proof}
 
\begin{proof}[Proof of Theorem~\ref{thm:output-space-limit}]
By Lemma~\ref{lem:hausdorff-support}, it suffices to show that, almost
surely,
\begin{equation}
    \sup_{H\in\mathcal B}\bigl|h_n(H)-h_{\mathscr K_{k,t}}(H)\bigr|
    \longrightarrow0.
    \label{eq:support-uniform}
\end{equation}
We divide the proof into four steps.
 
\medskip
\noindent
\textbf{Step 1: the limiting support function.}
Let $H\in\mb M_k^{\rm sa}$ have eigenvalues $h_1,\ldots,h_k$.  By von
Neumann's trace inequality and the permutation invariance of
$\Lambda_{k,t}$,
\begin{equation}
    h_{\mathscr K_{k,t}}(H)
    =
    \max_{u\in\mathscr D_{k,t}}
    \frac{\sum_ih_iu_i}{\sum_iu_i}.
    \label{eq:limit-support-formula}
\end{equation}
 
\medskip
\noindent
\textbf{Step 2: diagonal $H$.}
Let $H=\diag(h)$ with $h\in\mb R^k$, and write $\mathbf 1:=(1,\ldots,1)$.
A partial-trace calculation gives $T_n(\diag(c))=S_n(c)$ for
$c\in\mb R^k$, where $S_n(c)=\sum_ic_iP_{ii}^{(n)}$ is as in the proof of
Lemma~\ref{lem:random-matrix-input}.  Hence, by
Lemma~\ref{lem:support-channel},
\[
h_n(H)=\inf\{\lambda\in\mb R:f_n(\lambda)\le0\},\quad
f_n(\lambda):=\lambda_{\max}(S_n(h-\lambda\mathbf 1)).
\]
Since $S_n(h-\lambda\mathbf 1)=S_n(h)-\lambda P_{A,n}$ and $P_{A,n}\ge0$, the
function $f_n$ is nonincreasing.  Put
$$f(\lambda):=\max_{u\in\mathscr D_{k,t}}\sum_i(h_i-\lambda)u_i$$ and let
$\lambda_*:=h_{\mathscr K_{k,t}}(H)$, the maximum in
\eqref{eq:limit-support-formula} being attained at
$u^*\in\mathscr D_{k,t}$.  Since
$\sum_ih_iu_i\le\lambda_*\sum_iu_i$ on $\mathscr D_{k,t}$, with equality
at $u^*$, every $\varepsilon>0$ satisfies
$f(\lambda_*+\varepsilon)\le-\varepsilon m_{k,t}<0$ and
$f(\lambda_*-\varepsilon)\ge\varepsilon\sum_iu^*_i>0$.  By
Lemma~\ref{lem:random-matrix-input}, almost surely,
$f_n(\lambda)\to f(\lambda)$ simultaneously for all $\lambda\in\mb R$
and $h\in\mb R^k$.  Hence
$f_n(\lambda_*+\varepsilon)<0<f_n(\lambda_*-\varepsilon)$ for all
sufficiently large $n$, and the monotonicity of $f_n$ gives
$\lambda_*-\varepsilon\le h_n(H)\le\lambda_*+\varepsilon$.  Therefore,
almost surely, $h_n(H)\to h_{\mathscr K_{k,t}}(H)$ for every diagonal
$H\in\mb M_k^{\rm sa}$.
 
\medskip
\noindent
\textbf{Step 3: general $H$.}
Let $H=V\diag(h)V^*$ with $V\in\mathcal U(k)$ fixed, and put
$Q_n:=(I_n\otimes V^*)P_n(I_n\otimes V)$.  The projections $Q_n$ are
again independent and Haar distributed with rank $d_n$.  Moreover, the
cyclicity of the partial trace gives $\Tr_BQ_n=P_{A,n}$ and
$T_n^Q(\diag(h))=T_n^P(H)$, where $T_n^Q$ and $T_n^P$ denote
\eqref{eq:def-T-map} for $Q_n$ and $P_n$, respectively.  Applying Step~2
to $(Q_n)$, we obtain $h_n(H)\to h_{\mathscr K_{k,t}}(H)$ almost surely,
for each fixed $H\in\mb M_k^{\rm sa}$.
 
\medskip
\noindent
\textbf{Step 4: uniformity in $H$.}
Let $\mathcal N$ be a countable dense subset of $\mathcal B$ defined in~\eqref{eq:unit-ball-operator-norm}, and
intersect the probability-one events of Step~3 over $H\in\mathcal N$.
Support functions of subsets of $\mc D_k$ are $1$-Lipschitz in the
operator norm, since $|\Tr((H-H')\sigma)|\le\|H-H'\|_\infty$ for
$\sigma\in\mc D_k$.  Since $\mathcal B$ is compact, for every
$\delta>0$ it is covered by finitely many $\delta$-balls centered in
$\mathcal N$.  Hence the supremum in \eqref{eq:support-uniform} is at
most the maximum error over these centers plus $2\delta$.  Letting first
$n\to\infty$ and then $\delta\downarrow0$ proves
\eqref{eq:support-uniform}.
\end{proof}

\section{Bell output of the product--conjugate channel}
\label{sec:bell-output}
Let \(P_n=P_{AB,n}\) and \(\Phi_n\) be as in
Section~\ref{sec:output-space}, and put
\begin{equation}
    Z_n:=(\Phi_n\otimes\overline{\Phi_n})(\psi_n^+)
    \in\mc D_{k^2}.
\end{equation}
Here the Bell input has dimension \(n\); we write \(\psi_k^+\) for
the Bell state on the output space.

\begin{theorem}[Limit of the Bell output]
\label{thm:bell-output-limit}
Fix \(k\geq2\), and suppose that
\begin{equation}
    \frac{d_n}{nk}\longrightarrow t\in(k^{-2},1).
\end{equation}
Then, almost surely,
\begin{equation}
    \left\|Z_n-
      \left[r_{k,t}\psi_k^++(1-r_{k,t})\frac{I_{k^2}}{k^2}\right]
    \right\|_\infty\longrightarrow0,
    \label{eq:bell-output-limit}
\end{equation}
where
\begin{equation}
    r_{k,t}
    =\frac{k^2(1-t)}{k^4t-2k^2t+1}
    =\frac{k^2(1-t)}{(k^2-1)^2t+1-t}.
    \label{eq:r-kt}
\end{equation}
\end{theorem}

\begin{proof}
Write the normalized Choi matrix as
\begin{equation}
    J_n=(P_{A,n}^{-1/2}\otimes I_k)P_n
             (P_{A,n}^{-1/2}\otimes I_k)
       =\sum_{i,j=1}^kJ_{ij}^{(n)}\otimes|i\rangle\langle j|.
\end{equation}
The Choi formula and \(J_{ji}^{(n)}=(J_{ij}^{(n)})^*\) give
\begin{equation}
    (Z_n)_{(i,p),(j,q)}
    =\frac1n\Tr\bigl(J_{ij}^{(n)}J_{qp}^{(n)}\bigr).
    \label{eq:Z-block-entry}
\end{equation}
By Proposition~\ref{prop:normalized-choi-block-limit}, these entries
converge almost surely to
\begin{equation}
    \beta_{k,t}\delta_{ij}\delta_{pq}
       +\frac{r_{k,t}}k\delta_{ip}\delta_{jq},
    \qquad
    \beta_{k,t}:=\frac{1-r_{k,t}}{k^2}.
    \label{eq:bell-entry-limit}
\end{equation}
This is the matrix of \(\beta_{k,t}I_{k^2}+r_{k,t}\psi_k^+\).
Since \(k\) is fixed, entrywise convergence gives operator-norm
convergence.
\end{proof}

The limiting eigenvalues are
\begin{align}
    \alpha_{k,t}
       &=\beta_{k,t}+r_{k,t}
        =\frac{k^4-(1+t)k^2+1}{k^2(k^4t-2k^2t+1)},
        \label{eq:alpha-kt}\\
    \beta_{k,t}
       &=\frac{(k^2-1)(k^2t-1)}{k^2(k^4t-2k^2t+1)}.
        \label{eq:beta-kt}
\end{align}
Thus, writing the spectrum with multiplicities,
\begin{equation}
    \lambda^{\mathrm{Bell}}_{k,t}
       :=\bigl(\alpha_{k,t},
          \underbrace{\beta_{k,t},\ldots,\beta_{k,t}}_{k^2-1}\bigr).
    \label{eq:def-lambda-bell}
\end{equation}

\begin{corollary}
\label{cor:joint-asymptotic}
For fixed \(k\geq2\), \(t\in(k^{-2},1)\), and \(p>0\), almost surely,
\begin{equation}
    \limsup_{n\to\infty}S_p^{\min}
       (\Phi_n\otimes\overline{\Phi_n})
    \leq S_p(\lambda^{\mathrm{Bell}}_{k,t}).
    \label{eq:joint-fixed-k-bound}
\end{equation}
For fixed \(t\in(0,1)\) and \(p>0\), as \(k\to\infty\),
\begin{equation}
    S_p(\lambda^{\mathrm{Bell}}_{k,t})
    =2\log k-\frac{A_p(t)}{k^2}+O_{p,t}(k^{-4}),
    \label{eq:bell-entropy-large-k}
\end{equation}
where
\begin{align}
    A_p(t)&:=\frac{t^{-p}-1-p(t^{-1}-1)}{p-1},
      \quad p\ne1,\label{eq:def-Apt}
\end{align}
and $A_1(t):=t^{-1}\log(t^{-1})-t^{-1}+1.$
\end{corollary}

\begin{proof}
The Bell input gives
\(S_p^{\min}(\Phi_n\otimes\overline{\Phi_n})\leq S_p(Z_n)\).
The first assertion follows from the theorem and continuity of
\(S_p\) in fixed dimension.

For the expansion, put \(m=k^2\), \(s=t^{-1}\),
\(x=m\alpha_{k,t}\), and \(y=m\beta_{k,t}\). Then
\begin{equation}
    x=s+O_t(m^{-1}),\qquad
    y=1-\frac{s-1}{m}+O_t(m^{-2}).
\end{equation}
For \(p\ne1\), Taylor expansion gives
\begin{equation}
    x^p+(m-1)y^p
      =m+s^p-1-p(s-1)+O_{p,t}(m^{-1}).
\end{equation}
Taking the logarithm proves \eqref{eq:bell-entropy-large-k}.
For \(p=1\), use
\begin{equation}
    x\log x+(m-1)y\log y
       =s\log s-s+1+O_t(m^{-1}).
\end{equation}
The coefficient at \(p=1\) is also the continuous extension of
\(A_p(t)\).
\end{proof}


\section{Antisymmetric postprocessing}\label{sec:antisymmetric-postprocessing}
 
In this section, we study the asymptotic entropies after antisymmetric
postprocessing.  Let \(P_n\) and \(\Phi_n\) be as in
Section~\ref{sec:output-space}, let \(1\le r\le k\), and recall the
postprocessing channel \(\mc E_{k,r}\) from
\eqref{def:channel-postprocessing}.  We put
\begin{equation}
    \Theta_n:=\mc E_{k,r}\circ\Phi_n.
    \label{eq:channel-postprocessing}
\end{equation}
The channel \(\mc E_{k,r}\) is known in the literature.  It appeared as
the normalized \(r\)-particle block of the Grassmann channel
in~\cite{BradlerJochymOConnorJauregui2011}, and as the \(1\to r\)
antisymmetric cloning channel obtained by specializing the
subspace-cloning construction of \cite[Eq.~(15)]{LemmWilde2017}.  For
\(r=2\), a complementary channel of \(\mc E_{k,2}\) is the
Werner--Holevo channel introduced in \cite[Eq.~(6)]{WernerHolevo2002}:
\begin{equation}
    \mc W_k(X)
    :=
    \frac{\Tr(X)I_k-X^T}{k-1}.
    \label{eq:Werner-Holevo-channel}
\end{equation}
The main result of this section is the following.
 
\begin{theorem}[Single-output and Bell-output asymptotics]
\label{thm:entropy-asymptotics-antisymmetric}
Fix an integer \(r\geq1\), \(t\in(0,1)\), and \(p\in(0,\infty)\), and put
\(\gamma:=(1-t)/t\) and \(D_{k,r}:=\binom kr\).  For every \(k\ge r\) with
\(k^2t>1\), the limits below exist almost surely, and as \(k\to\infty\),
\begin{align}
 \lim_{n\to\infty}S_p^{\min}(\Theta_n)
 &=\log D_{k,r}-\frac{2p\gamma}{rk^2}+o(k^{-2}),
 \label{eq:single-asymptotic}\\
 \lim_{n\to\infty}S_p\!\left(
   (\Theta_n\otimes\overline{\Theta}_n)(\psi_n^+)\right)
 &=2\log D_{k,r}-\frac{B_{p,r}(\gamma)}{k^2}+o(k^{-2}),
 \label{eq:joint-asymptotic}
\end{align}
where, for \(p\neq1\),
\begin{equation}
 B_{p,r}(\gamma)
 :=\frac{p\gamma-r^2\left[
       (1+\gamma/r^2)^p-1\right]}{1-p},
 \label{eq:Bpr}
\end{equation}
and its continuous extension at \(p=1\) is
\begin{equation}
 B_{1,r}(\gamma)
 =(r^2+\gamma)\log\left(1+\frac{\gamma}{r^2}\right)-\gamma.
 \label{eq:B1r}
\end{equation}
\end{theorem}
 
For \(r=1\), the channel \(\mc E_{k,1}\) is the identity and
\(B_{p,1}(\gamma)=A_p(t)\), so \eqref{eq:single-asymptotic} and
\eqref{eq:joint-asymptotic} reduce to \eqref{eq:single-large-k} and
\eqref{eq:intro-joint-asymptotic}.  The rest of this section proves
Theorem~\ref{thm:entropy-asymptotics-antisymmetric}.

We first talk about the basic properties of $\mc E_{k,r}$. Fix \(1\leq r\leq k\), and put
\begin{equation}
 D_{k,r}:=\binom{k}{r},
 \qquad
 N_{k,r}:=\binom{k-1}{r-1},
 \qquad
 kN_{k,r}=rD_{k,r}.
 \label{eq:DN}
\end{equation}
Let \(\mathfrak S_r\) be the symmetric group on \(r\) letters.  For
\(\pi\in\mathfrak S_r\), let \(P_\pi\) be the permutation unitary on
\((\mb C^k)^{\otimes r}\) defined on the computational basis by
\[
 P_\pi|j_1\cdots j_r\rangle
 =|j_{\pi^{-1}(1)}\cdots j_{\pi^{-1}(r)}\rangle,
 \qquad j_1,\ldots,j_r\in[k].
\]
The antisymmetrizer
\begin{equation}
 \Pi_{k,r}^{\asym}
 :=\frac1{r!}\sum_{\pi\in\mathfrak S_r}\operatorname{sgn}(\pi)P_\pi
 \label{eq:antisymmetrizer}
\end{equation}
is the orthogonal projection onto the antisymmetric subspace
\[
 \mc A_{k,r}^{\asym}
 =\left\{\xi\in(\mb C^k)^{\otimes r}:
 P_\sigma\xi=\operatorname{sgn}(\sigma)\xi
 \text{ for every }\sigma\in\mathfrak S_r\right\}.
\]
For \(X\in\mb M_k\) and \(1\le j\le r\), write
\[
 X^{(j)}:=I_k^{\otimes(j-1)}\otimes X\otimes
 I_k^{\otimes(r-j)}.
\]
For \(2\leq j\leq r\), the transposition \(\tau_j=(1\ j)\) satisfies
\[
 X^{(j)}=P_{\tau_j}X^{(1)}P_{\tau_j},
 \qquad
 \Pi_{k,r}^{\asym}P_{\tau_j}
 =P_{\tau_j}\Pi_{k,r}^{\asym}
 =-\Pi_{k,r}^{\asym},
\]
so the two signs cancel after compression by \(\Pi_{k,r}^{\asym}\).
Hence, since \(k/D_{k,r}=r/N_{k,r}\), the channel
\eqref{def:channel-postprocessing} has the one-leg and symmetric forms
\begin{align}
 \mc E_{k,r}(X)
 &=\frac{k}{D_{k,r}}\Pi_{k,r}^{\asym}
     \bigl(X\otimes I_k^{\otimes(r-1)}\bigr)\Pi_{k,r}^{\asym}
 \label{eq:E-one-leg}\\
 &=\frac1{N_{k,r}}\Pi_{k,r}^{\asym}
       \Bigl(\sum_{j=1}^rX^{(j)}\Bigr)
       \Pi_{k,r}^{\asym}.
 \label{eq:E-definition}
\end{align}
 
For \(I=\{i_1<\cdots<i_r\}\subseteq[k]\), the normalized Slater vector
\begin{equation}
 |e_I^{\asym}\rangle
 :=\frac1{\sqrt{r!}}\sum_{\pi\in\mathfrak S_r}
     \operatorname{sgn}(\pi)|i_{\pi(1)}\cdots i_{\pi(r)}\rangle
 \label{eq:Slater}
\end{equation}
lies in \(\mc A_{k,r}^{\asym}\), and the vectors \(|e_I^{\asym}\rangle\),
\(|I|=r\), form an orthonormal basis of \(\mc A_{k,r}^{\asym}\).
Consequently,
\begin{equation}
 \Pi_{k,r}^{\asym}
 =\sum_{\substack{I\subseteq[k]\\|I|=r}}
   |e_I^{\asym}\rangle\langle e_I^{\asym}|,
 \qquad
 \dim\mc A_{k,r}^{\asym}=D_{k,r}.
 \label{eq:Slater-resolution}
\end{equation}
The one-system reduced state of \(|e_I^{\asym}\rangle\) is
\(\frac1r\sum_{i\in I}|i\rangle\langle i|\).  Hence, if
\(\Tr_{\widehat j}\) traces out all factors except the \(j\)th one,
\begin{equation}
 \Tr_{\widehat j}\Pi_{k,r}^{\asym}
 =\frac1r\sum_{|I|=r}\sum_{i\in I}|i\rangle\langle i|
 =\frac{N_{k,r}}rI_k.
 \label{eq:one-body-marginal}
\end{equation}
By \eqref{eq:E-one-leg} and \eqref{eq:DN}, this gives
\(\Tr\mc E_{k,r}(X)=\frac{k}{D_{k,r}}\Tr\bigl(X\Tr_{\widehat1}\Pi_{k,r}^{\asym}\bigr)=\Tr X\).
Thus \(\mc E_{k,r}\) is a quantum channel from \(\mb M_k\) to
\(\mc B(\mc A_{k,r}^{\asym})\).  Since \(\Pi_{k,r}^{\asym}\) is a real
matrix, \(\overline{\mc E}_{k,r}=\mc E_{k,r}\).
 
\begin{proposition}[Eigenvalue shuffling]
\label{prop:eigenvalue-shuffling}
Let \(X\in\mb M_k^{\rm sa}\) have spectral decomposition
\[
 X=\sum_{i=1}^k\lambda_i|v_i\rangle\langle v_i|.
\]
For every \(r\)-subset \(I\subseteq[k]\), let
\(|v_I^{\asym}\rangle\) be the Slater vector formed from
\(\{v_i:i\in I\}\) as in \eqref{eq:Slater}.  Then
\begin{equation}
 \mc E_{k,r}(X)|v_I^{\asym}\rangle
 =\frac1{N_{k,r}}\Bigl(\sum_{i\in I}\lambda_i\Bigr)
   |v_I^{\asym}\rangle.
 \label{eq:eigenvalue-shuffling}
\end{equation}
Thus the output spectrum, as a multiset indexed by the \(r\)-subsets, is
\begin{equation}
 \operatorname{spec}(\mc E_{k,r}(X))
 =\left\{\frac1{N_{k,r}}\sum_{i\in I}\lambda_i:
 I\subseteq[k],\ |I|=r\right\}.
 \label{eq:E-spectrum}
\end{equation}
\end{proposition}
 
\begin{proof}
Every elementary tensor in \(|v_I^{\asym}\rangle\) contains each \(v_i\),
\(i\in I\), exactly once.  Therefore
\[
 \left(\sum_{j=1}^rX^{(j)}\right)|v_I^{\asym}\rangle
 =\left(\sum_{i\in I}\lambda_i\right)|v_I^{\asym}\rangle.
\]
If \(U\) is the unitary with \(U|i\rangle=|v_i\rangle\), then
\(|v_I^{\asym}\rangle=U^{\otimes r}|e_I^{\asym}\rangle\), and \(U^{\otimes r}\)
commutes with \(\Pi_{k,r}^{\asym}\).  Hence the vectors
\(|v_I^{\asym}\rangle\) form an orthonormal basis of
\(\mc A_{k,r}^{\asym}\), and the conclusion follows from
\eqref{eq:E-definition}.
\end{proof}
Finally, we prove the main result in this section, with the technical details in the appendix:
\begin{proof}[Proof of Theorem~\ref{thm:entropy-asymptotics-antisymmetric}]
    Let \(k\ge r\) with \(k^2t>1\), so that the results of
Section~\ref{sec:output-space} apply.  Since \(\mc E_{k,r}\) is a
trace-norm contraction,
\(d_H^{(1)}(\Theta_n(\mc D_n),\mc E_{k,r}(\mathscr K_{k,t}))\le
d_H^{(1)}(\Phi_n(\mc D_n),\mathscr K_{k,t})\).  Hence
Theorem~\ref{thm:output-space-limit} and the uniform continuity of
\(S_p\) on the density operators on \(\mc A_{k,r}^{\asym}\) give, almost
surely,
\[
    \lim_{n\to\infty}S_p^{\min}(\Theta_n)
    =\min_{\rho\in\mathscr K_{k,t}}S_p\bigl(\mc E_{k,r}(\rho)\bigr).
\]
Similarly,
\((\Theta_n\otimes\overline\Theta_n)(\psi_n^+)
=(\mc E_{k,r}\otimes\overline{\mc E}_{k,r})
\bigl((\Phi_n\otimes\overline\Phi_n)(\psi_n^+)\bigr)\), so
Theorem~\ref{thm:bell-output-limit} gives, almost surely,
\[
    \lim_{n\to\infty}S_p\bigl((\Theta_n\otimes\overline\Theta_n)(\psi_n^+)\bigr)
    =S_p\bigl((\mc E_{k,r}\otimes\overline{\mc E}_{k,r})(Z_{k,t})\bigr),
\]
where \(Z_{k,t}=r_{k,t}\psi_k^++(1-r_{k,t})I_{k^2}/k^2\) is the isotropic
state of Theorem~\ref{thm:bell-output-limit}.  The large-\(k\)
asymptotics of the two right-hand sides are given by
Propositions~\ref{prop:app-antisymmetric-single-entropy-asymptotics}
and~\ref{prop:app-antisymmetric-bell-entropy-asymptotics} in
Appendix~\ref{app:entropy-asymptotics}, with error terms
\(O(k^{-5/2})\) and \(O(k^{-3})\), respectively.  This proves
Theorem~\ref{thm:entropy-asymptotics-antisymmetric}.
\end{proof}

\section{Proof of the main result}
\label{sec:main-proof}
In this section, we present the proof of the main result. We start by proving nonadditivity for minimum output
\(p\)-R\'enyi entropy for all $p>0$. Then we present the numerical calculation for nonadditivity of minimum output von Neumann entropy with an output dimension $182$, improving the threshold established in \cite{Belinschi_2016}. 

\subsection{Proof of nonadditivity of minimum output
\(p\)-R\'enyi entropy for all $p>0$}

\begin{proof}[Proof of Theorem~\ref{thm:main}]
Fix \(p>0\), and put \(r=5\) and \(\gamma=375\), so that
\(t=(1+\gamma)^{-1}=1/376\). Recall the formula of $B_{p,r}(\gamma)$ in Theorem~\ref{thm:entropy-asymptotics-antisymmetric}, then Taylor's formula with integral remainder
gives, with the continuous interpretation at \(p=1\), 
\begin{equation}
    \frac{B_{p,5}(375)}{375p}
    =\frac1{15}\int_1^{16}(16-s)s^{p-2}\,ds
    \geq\frac1{15}\int_1^{16}\frac{16-s}{s^2}\,ds
    =1-\frac{\log16}{15}
    >\frac45.
    \label{eq:main-coefficient-gap}
\end{equation}
Here we used \(s^p\geq1\) for \(s\geq1\) and \(\log16<3\).
Thus
\begin{equation}
    \Delta_p:=B_{p,5}(375)-300p
       \geq25p(3-\log16)>0.
\end{equation}
By Theorem~\ref{thm:entropy-asymptotics-antisymmetric}, almost surely,
\begin{equation}
    \lim_{n\to\infty}
    \left[
        2S_p^{\min}(\Theta_n)
        -S_p\!\left(
          (\Theta_n\otimes\overline \Theta_n)(\psi_n^+)
        \right)
    \right]
    =\frac{\Delta_p}{k^2}+o(k^{-2}).
    \label{eq:main-asymptotic-gap}
\end{equation}
Choose \(k\) sufficiently large that the right-hand side is positive
and \(k\geq20\). Then \(k\geq r\) and \(t>k^{-2}\), so the channels
are well defined almost surely for all sufficiently large \(n\).
Almost surely, increasing \(n\) further makes the expression inside
the limit strictly positive.
Since complex conjugation preserves minimum output entropy,
\begin{align}
    S_p^{\min}(\Theta_n\otimes\overline \Theta_n)
    &\leq S_p\!\left(
          (\Theta_n\otimes\overline \Theta_n)(\psi_n^+)
        \right)\notag\\
    &<2S_p^{\min}(\Theta_n)
      =S_p^{\min}(\Theta_n)+S_p^{\min}(\overline \Theta_n).
\end{align}
Taking \(\Phi=\Theta_n\) and \(\Psi=\overline \Theta_n\) proves the theorem.
\end{proof}
The same argument applies to any fixed integer \(r\geq1\) whenever
\(B_{p,r}(\gamma)>4p\gamma/r\). Indeed, setting
\(x=1+\gamma/r^2\), we have
\begin{equation}
    \frac{B_{p,r}(\gamma)}{p\gamma}
    =\frac1{x-1}\int_1^x(x-s)s^{p-2}\,ds
    \longrightarrow
    \begin{cases}
        (1-p)^{-1},&0<p<1,\\
        +\infty,&p\geq1,
    \end{cases}
    \qquad \gamma\to\infty.
    \label{eq:antisymmetric-threshold-limit}
\end{equation}
Consequently, choosing \(\gamma\) sufficiently large proves
nonadditivity for
\(p>\max\{1-r/4,0\}\), giving the ranges
\begin{equation}
    \begin{array}{c|c}
        r&\text{R\'enyi orders obtained}\\ \hline
        1&p>3/4\\
        2&p>1/2\\
        3&p>1/4\\
        r\geq4&p>0.
    \end{array}
\end{equation}
Thus \(r=4\) already suffices for every fixed \(p>0\), with
\(\gamma\) chosen depending on \(p\). The choice \(r=5\) above
also makes \(\gamma\) independent of \(p\).

\subsection{Improved output dimension threshold}
 
For \(p>0\), define
\begin{equation}
 k_{\mathrm{high}}(p)
 :=\min\left\{k\geq2:\ \exists t\in(k^{-2},1),\quad
 S_p(\lambda^{\mathrm{Bell}}_{k,t})
 <2\min_{\sigma\in\mathscr K_{k,t}}S_p(\sigma)\right\},
 \label{eq:k-high-def}
\end{equation}
with \(\min\varnothing=+\infty\). 
 
\begin{proposition}
\label{prop:certified-k182}
For \(k=182\) and \(t=27/100000\),
\begin{equation}
 2\min_{\sigma\in\mathscr K_{182,t}}S_1(\sigma)
 -S_1(\lambda^{\mathrm{Bell}}_{182,t})
 >\frac{477}{10^6}.
 \label{eq:certified-k182-gap}
\end{equation}
In particular, \(k_{\mathrm{high}}(1)\leq182\).
\end{proposition}
 
\begin{proof}
We bound the two terms in \eqref{eq:certified-k182-gap} separately. All numerical
enclosures below are certified by outward-rounded interval arithmetic; see
Appendix~\ref{app:k182-certificate}.
 
\medskip\noindent\textbf{Step 1: lower bound for one copy.}
Put
\[
 h_k(x):=-x\log x-(1-x)\log\frac{1-x}{k-1},\qquad
 L:=\frac{162513}{10^6},\qquad z:=\frac{2077}{2000}.
\]
The hypotheses of Lemma~\ref{lem:finite-k-entropy-reduction} hold for \(k=182\) and
\(t=27/100000\). By that lemma, the minimum of \(S_1\) over \(\mathscr K_{182,t}\) equals
\(h_{182}(L^*_{182,t})\), where \(L^*_{182,t}\) is the largest eigenvalue occurring in
\(\Lambda_{182,t}\), and \(h_{182}\) is decreasing on \([1/182,1)\). Moreover, Lemma~\ref{lem:finite-k-scalar-certificate} gives
\(L^*_{182,t}\leq L\). Hence
\[
 \min_{\sigma\in\mathscr K_{182,t}}S_1(\sigma)\geq h_{182}(L)>4.79748729461169.
\]
 
\medskip\noindent\textbf{Step 2: upper bound for two copies.}
By Theorem~\ref{thm:bell-output-limit} and its proof, \(\lambda^{\mathrm{Bell}}_{k,t}\)
consists of
\[
 \alpha_{k,t}=\frac{k^4-(1+t)k^2+1}{k^2(k^4t-2k^2t+1)},\qquad
 \beta_{k,t}=\frac{(k^2-1)(k^2t-1)}{k^2(k^4t-2k^2t+1)},
\]
with multiplicities \(1\) and \(k^2-1\). Hence
\[
 S_1(\lambda^{\mathrm{Bell}}_{k,t})
 =-\alpha_{k,t}\log\alpha_{k,t}-(k^2-1)\beta_{k,t}\log\beta_{k,t}.
\]
For \(k=182\) and \(t=27/100000\), \(\alpha_{182,t}=0.1118162961\ldots\) and
\(\beta_{182,t}=2.6814711947\ldots\times10^{-5}\) are explicit rationals, and
\[
 S_1(\lambda^{\mathrm{Bell}}_{182,t})<9.59449702812109.
\]
 
\medskip\noindent\textbf{Step 3: comparison.}
Steps~1 and~2 give
\[
 2\min_{\sigma\in\mathscr K_{182,t}}S_1(\sigma)-S_1(\lambda^{\mathrm{Bell}}_{182,t})
 >2\cdot4.79748729461169-9.59449702812109=0.00047756110229>\frac{477}{10^6}.
\]
Since \(t>182^{-2}\), the pair \((182,t)\) is admissible in \eqref{eq:k-high-def}, and
\(k_{\mathrm{high}}(1)\leq182\).
\end{proof}
 
Together with the output-set and Bell-output limits, this gives, for output dimension
\(182\), finite-dimensional channel pairs that violate additivity of the minimum output von
Neumann entropy. For comparison, the random-Stinespring/Bell-state criterion of
Belinschi, Collins, and Nechita first detects a violation at \(k=183\)~\cite{Belinschi_2016}.

\appendices

\section{Free-probability tools for random projection channels}
\label{app:free-probability}
 
This section collects the free-probability input used in
Lemma~\ref{lem:random-matrix-input} and in the Bell-output calculation
in Theorem~\ref{thm:bell-output-limit}.  Throughout, \(k\) is fixed and
\(n\to\infty\).  The large-\(k\) limit is taken only in
Appendix~\ref{app:entropy-asymptotics}.  For \(m\in\mathbb N\), we write
\(\tr_m:=\frac1m\Tr\) for the normalized trace on \(M_m(\mathbb C)\).
For general background, see
Refs.~\cite{MingoSpeicher2017,CollinsNechita2016Review,aubrun2017alice}.
 
\subsection{Necessary free-probability basics}
\label{app:sub-basics}
 
\paragraph{Weak and strong convergence}
 
Let \(X_n=X_n^*\in M_{m_n}(\mathbb C)\), where \(m_n\to\infty\).  The
empirical spectral distribution of \(X_n\) is
\(\mu_{X_n}:=m_n^{-1}\sum_{j=1}^{m_n}\delta_{\lambda_j(X_n)}\).  Thus
\(\tr_{m_n}(f(X_n))=\int_{\mathbb R}f\,d\mu_{X_n}\) for every bounded
Borel function \(f\colon\mathbb R\to\mathbb C\).  We say that
\(\mu_{X_n}\) converges weakly to a probability measure \(\mu\), and
write \(\mu_{X_n}\xrightarrow{\mathrm w}\mu\), if
\(\int_{\mathbb R}f\,d\mu_{X_n}\to\int_{\mathbb R}f\,d\mu\) for every
bounded continuous function \(f\colon\mathbb R\to\mathbb C\).
 
Suppose that the spectra of \(X_n\) are contained in a common compact
interval \(K\subset\mathbb R\), and that \(\mu\) is supported in \(K\).
Then the Stone--Weierstrass theorem implies that weak convergence is
equivalent to \(\tr_{m_n}(q(X_n))\to\int_Kq\,d\mu\) for every
polynomial \(q\in\mathbb C[x]\).
 
Assume now that \(\mu\) is compactly supported.  We say that \(X_n\)
converges \emph{strongly in distribution} to \(\mu\) if, for every
polynomial \(q\in\mathbb C[x]\),
\begin{align}
    \tr_{m_n}\bigl(q(X_n)\bigr)
    &\longrightarrow
    \int_{\mathbb R}q(x)\,d\mu(x),
    \label{eq:app-moment-convergence}\\
    \bigl\|q(X_n)\bigr\|
    &\longrightarrow
    \max_{x\in\operatorname{supp}(\mu)}|q(x)|.
    \label{eq:app-norm-convergence}
\end{align}
The first condition controls the asymptotic proportion of eigenvalues,
whereas the second rules out spectral outliers.  In particular, strong
convergence implies
\(\lambda_{\max}(X_n)\to\max\operatorname{supp}(\mu)\) and
\(\lambda_{\min}(X_n)\to\min\operatorname{supp}(\mu)\).  When the
matrices are random, weak or strong convergence almost surely means that
there exists a single event of probability one on which all the
corresponding limits hold.
 
\paragraph{Noncommutative laws and free convolution}
 
A tracial noncommutative probability space is a unital
\(C^*\)-algebra \(\mathcal A\) equipped with a tracial state \(\tau\).
If \(x=x^*\in\mathcal A\), its law is the unique compactly supported
probability measure \(\mu_x\) satisfying
\(\tau(q(x))=\int_{\mathbb R}q\,d\mu_x\) for every polynomial \(q\).
Unital subalgebras \(\mathcal A_i\subseteq\mathcal A\) are said to be
free if \(\tau(x_1x_2\cdots x_r)=0\) whenever \(x_j\in\mathcal A_{i_j}\),
\(\tau(x_j)=0\), and \(i_j\neq i_{j+1}\).  If \(x\) and \(y\) are free
self-adjoint elements with laws \(\mu\) and \(\nu\), respectively, then
the law of \(x+y\) is denoted by \(\mu\boxplus\nu\).  We write
\(\mu^{\boxplus r}\) for the \(r\)-fold free convolution
\(\mu\boxplus\cdots\boxplus\mu\), and denote by \(D_c\mu\) the
push-forward of \(\mu\) under the map \(x\mapsto cx\).
 
\paragraph{The Cauchy transform and its inverse branch}
 
Let \(\mu\) be a compactly supported probability measure on
\(\mathbb R\).  Its Cauchy transform is
\[
G_\mu(z):=\int_{\mathbb R}(z-x)^{-1}\,d\mu(x),
\]
defined for
\(z\in\mathbb C\setminus\operatorname{supp}(\mu)\), and
\(G_\mu(z)=1/z+O(z^{-2})\) as \(z\to\infty\).
 
The notation \(G_\mu^{-1}\) requires some care because \(G_\mu\) is not
globally injective.  We define only the inverse branch associated with
\(z=\infty\).  To do this, put
\[
H_\mu(\zeta):=G_\mu(\zeta^{-1})=\int_{\mathbb R}\zeta(1-\zeta x)^{-1}\,d\mu(x)
\]
for \(\zeta\) sufficiently close to \(0\), and set \(H_\mu(0)=0\).  Then
\(H_\mu'(0)=1\), so the analytic inverse-function theorem gives a unique
analytic function \(L_\mu\), defined near \(0\), such that
\(H_\mu(L_\mu(w))=w\) and \(L_\mu(0)=0\).  For \(w\neq0\) sufficiently
close to \(0\), define \(K_\mu(w):=1/L_\mu(w)\).  Equivalently,
\(K_\mu\) is the unique local inverse branch satisfying
\(G_\mu(K_\mu(w))=w\) and \(K_\mu(w)=1/w+O(1)\) as \(w\to0\).  The
\(R\)-transform of \(\mu\) is defined by 
\[
R_\mu(w):=K_\mu(w)-1/w.
\]
The pole of
\(K_\mu\) at \(w=0\) is canceled, so \(R_\mu\) is analytic near \(0\).
 
There is also a useful real-variable interpretation.  Put
\(r_\mu:=\max\operatorname{supp}(\mu)\) and
\(\omega_\mu:=\lim_{x\downarrow r_\mu}G_\mu(x)\in(0,\infty]\).  For
\(x>r_\mu\), we have \(G_\mu(x)>0\) and
\(G_\mu'(x)=-\int_{\mathbb R}(x-s)^{-2}\,d\mu(s)<0\).  Hence \(G_\mu\)
maps \((r_\mu,\infty)\) bijectively onto \((0,\omega_\mu)\).  Its
ordinary real inverse on this interval agrees, for sufficiently small
\(w>0\), with the branch \(K_\mu(w)\) defined above.  Indeed, \(H_\mu\)
has real Taylor coefficients, and hence so does \(L_\mu\).  Thus, for
small \(w>0\), \(K_\mu(w)\) is real, lies in \((r_\mu,\infty)\), and
satisfies \(G_\mu(K_\mu(w))=w\).
 
The two important \(R\)-transform identities used below are
\begin{equation}
    R_{\mu\boxplus\nu}(w)
    =
    R_\mu(w)+R_\nu(w),
    \quad
    R_{D_c\mu}(w)
    =
    cR_\mu(cw).
    \label{eq:app-R-identities}
\end{equation}
The first identity expresses the fact that the \(R\)-transform
linearizes free additive convolution.
 
\paragraph{Application to Haar projections}
 
Let \(P_n\) be a Haar-distributed projection of rank \(d_n\) on
\(\mathbb C^n\otimes\mathbb C^k\), and suppose that
\(t_n:=d_n/(nk)\to t\in(0,1)\).  Since the eigenvalues of \(P_n\) are
\(0\) and \(1\),
\[
\mu_{P_n}=(1-t_n)\delta_0+t_n\delta_1\xrightarrow{\mathrm w}b_t:=(1-t)\delta_0+t\delta_1. 
\]
In fact, this convergence is strong. Indeed, for every polynomial
\(q\), we have \(q(P_n)=q(0)(I-P_n)+q(1)P_n\).  Therefore
\(\tr_{nk}(q(P_n))=(1-t_n)q(0)+t_nq(1)\to(1-t)q(0)+tq(1)\), and
\(\|q(P_n)\|=\max\{|q(0)|,|q(1)|\}\) for all sufficiently large \(n\).
 
Fix \(a=(a_1,\ldots,a_k)\in\mathbb R^k\), and let \(A\) be the
diagonal matrix with diagonal entries \(a_1,\ldots,a_k\).  Consider the
Hermiticity-preserving map \(\varphi_a\) on \(M_k(\mathbb C)\) given
by \(\varphi_a(X):=\Tr(AX)I_k\).  This map need not be
completely positive when some \(a_i\) are negative.  For the
block-modification theorem~\cite{Nechita2018}, we use the output-first
Choi convention
\(C_\varphi:=\sum_{i,j=1}^k\varphi(E_{ij})\otimes E_{ij}\).  Since
\(\varphi_a(E_{ij})=a_i\delta_{ij}I_k\), we have
\(C_{\varphi_a}=I_k\otimes A\).  If \(\Pi_\lambda\) is the full spectral
projection of \(A\) associated with an eigenvalue \(\lambda\), then the
corresponding spectral projection of \(C_{\varphi_a}\) is
\(I_k\otimes\Pi_\lambda\), and
\((\operatorname{id}_k\otimes\Tr)(I_k\otimes\Pi_\lambda)
=\operatorname{rank}(\Pi_\lambda)I_k\).  Thus \(\varphi_a\) satisfies
the unitarity condition required by the block-modification theorem~\cite[Theorem 5.1]{Nechita2018}. 
 
Write \(P_n=\sum_{i,j=1}^kP_{ij}^{(n)}\otimes E_{ij}\) as a
\(k\times k\) block matrix with blocks
\(P_{ij}^{(n)}\in M_n(\mathbb C)\), and define
\(S_n(a):=\sum_{i=1}^ka_iP_{ii}^{(n)}\).  Since
\(\varphi_a(E_{ij})=a_i\delta_{ij}I_k\), a direct calculation gives
\((\operatorname{id}_n\otimes\varphi_a)(P_n)=S_n(a)\otimes I_k\).
The block-modification theorem
\cite[Theorem~5.1]{ArizmendiNechitaVargas2016}, together with its strong
version \cite[Section~5]{Nechita2018}, implies that, for every fixed
\(a\in\mathbb R^k\), the matrices \(S_n(a)\otimes I_k\) converge
strongly, almost surely, to the probability measure
\begin{equation}
    \mu_{a,t}
    :=
    \mathop{\boxplus}_{i=1}^k
    \left(D_{a_i/k}b_t\right)^{\boxplus k}.
    \label{eq:app-compression-law-measure}
\end{equation}
Every eigenvalue of \(S_n(a)\) occurs exactly \(k\) times in
\(S_n(a)\otimes I_k\).  Consequently,
\(\mu_{S_n(a)\otimes I_k}=\mu_{S_n(a)}\) and
\(\|q(S_n(a)\otimes I_k)\|=\|q(S_n(a))\|\), so strong convergence
yields, almost surely,
\begin{equation}
    \lambda_{\max}(S_n(a))
    \longrightarrow
    \max\operatorname{supp}(\mu_{a,t}).
    \label{eq:app-compression-law}
\end{equation}
 
\subsection{Proof of the random compression formula}
\label{app:proof-random-compression}
 
In this subsection, we compute \(\max\operatorname{supp}(\mu_{a,t})\),
thus concluding the proof of Lemma~\ref{lem:random-matrix-input}.
 
\paragraph{The Bernoulli \(R\)-transform}
 
The Cauchy transform of \(b_t=(1-t)\delta_0+t\delta_1\) is
\[
G_{b_t}(z)=(1-t)/z+t/(z-1)=(z-(1-t))/(z(z-1)).
\]
Setting \(w=G_{b_t}(z)\) gives \(wz^2-(1+w)z+(1-t)=0\), and hence
\(z=\bigl(1+w\pm\sqrt{(1-w)^2+4tw}\bigr)/(2w)\).  The inverse branch
corresponding to \(z=\infty\) is obtained by taking the plus sign, since
this is the unique branch satisfying \(z\sim w^{-1}\) as \(w\to0\).  Hence
\(K_{b_t}(w)=\bigl(1+w+\sqrt{(1-w)^2+4tw}\bigr)/(2w)\), and therefore
\begin{equation}
    R_t(w):=R_{b_t}(w)
    =
    \frac{w-1+\sqrt{(1-w)^2+4tw}}{2w}.
    \label{eq:app-bernoulli-R}
\end{equation}
In particular, \(R_t(0)=t\).  By \eqref{eq:app-R-identities},
\(R_{(D_c\nu)^{\boxplus k}}(w)=kcR_\nu(cw)\).  Applying this with
\(c=a_i/k\) to each factor in \eqref{eq:app-compression-law-measure}, we
obtain
\begin{equation}
    R_{\mu_{a,t}}(w)
    =
    \sum_{i=1}^ka_iR_t\left(\frac{a_iw}{k}\right).
    \label{eq:app-R-mu-a-t}
\end{equation}
Define
\begin{equation}
    g_t(y)
    :=
    yR_t(y)
    =
    \frac{y-1+\Delta_t(y)}{2},
    \label{eq:app-g-t}
\end{equation}
where \(\Delta_t(y):=\sqrt{(1-y)^2+4ty}\).  Since
\(\Delta_t(y)^2=(y+2t-1)^2+4t(1-t)\) and \(t\in(0,1)\), we have
\(\Delta_t(y)>0\) for every \(y\in\mathbb R\).
Thus \(g_t\) is real analytic on the entire real line.  Moreover,
\(\Delta_t(0)=1\), so \(g_t(0)=0\).
 
\paragraph{The elementary Legendre identity}
 
Recall that
\(c_t(u):=\bigl(\sqrt{t(1-u)}-\sqrt{(1-t)u}\bigr)^2\) for
\(0\leq u\leq1\).  Expanding the square gives
\(c_t(u)=t+(1-2t)u-2\sqrt{t(1-t)u(1-u)}\).  Hence, for \(0<u<1\),
\begin{equation}
    c_t'(u)
    =
    1-2t
    -
    \sqrt{t(1-t)}\,
    \frac{1-2u}{\sqrt{u(1-u)}},
    \label{eq:app-ct-derivative}
\end{equation}
and \(c_t''(u)=\sqrt{t(1-t)}\big/\bigl(2[u(1-u)]^{3/2}\bigr)>0\).  Thus
\(c_t\) is strictly convex, with \(c_t'(u)\to-\infty\) as
\(u\downarrow0\) and \(c_t'(u)\to+\infty\) as \(u\uparrow1\).  It follows
that, for every \(y\in\mathbb R\), the function \(u\mapsto yu-c_t(u)\)
has a unique maximizer on \([0,1]\), and that this maximizer lies in
\((0,1)\).
 
Differentiating \eqref{eq:app-g-t} gives
\begin{equation}
    u_t(y)
    :=
    g_t'(y)
    =
    \frac12
    \left(
      1+\frac{y+2t-1}{\Delta_t(y)}
    \right).
    \label{eq:app-ut}
\end{equation}
Because \(\Delta_t(y)^2-(y+2t-1)^2=4t(1-t)>0\), we have
\(0<u_t(y)<1\), \(1-2u_t(y)=-(y+2t-1)/\Delta_t(y)\), and
\(u_t(y)(1-u_t(y))=t(1-t)/\Delta_t(y)^2\).  Substituting these
identities into \eqref{eq:app-ct-derivative} gives \(c_t'(u_t(y))=y\).
Thus \(u_t(y)\) is the unique maximizer of \(u\mapsto yu-c_t(u)\).
 
Let \(h_t(y):=yu_t(y)-c_t(u_t(y))\).  Since \(c_t'(u_t(y))=y\), we have
\(h_t'=u_t=g_t'\).  Moreover, \(u_t(0)=t\) and \(c_t(t)=0\), so
\(h_t(0)=0=g_t(0)\).  Hence \(h_t=g_t\), and we have proved the Legendre
identity
\begin{equation}
    g_t(y)
    =
    \max_{0\leq u\leq1}
    \{yu-c_t(u)\},
    \qquad
    u_t(y)=g_t'(y).
    \label{eq:app-legendre-pair}
\end{equation}
Equivalently, for every \(y\in\mathbb R\) and \(u\in[0,1]\),
\begin{equation}
    yu-c_t(u)\leq g_t(y),
    \label{eq:app-legendre-inequality}
\end{equation}
with equality precisely when \(u=u_t(y)\).
 
Recall that
\begin{equation}
    \mathscr D_{k,t}
    =
    \left\{
      u\in[0,1]^k:
      \sum_{i=1}^kc_t(u_i)\leq\frac1k
    \right\}.
\end{equation}
 
\begin{lemma}\label{lem:bernoulli-edge-duality}
For every \(a=(a_1,\ldots,a_k)\in\mathbb R^k\),
\begin{equation}
    \max\operatorname{supp}(\mu_{a,t})
    =
    \max_{u\in\mathscr D_{k,t}}
    \sum_{i=1}^ka_i u_i.
    \label{eq:app-edge-support-duality}
\end{equation}
\end{lemma}
 
\begin{proof}
Since \(c_t\) is continuous, \(\mathscr D_{k,t}\) is a compact subset
of \([0,1]^k\), and it is nonempty because \(c_t(t)=0\).  Hence the
maximum \(M_a:=\max_{u\in\mathscr D_{k,t}}\sum_{i=1}^ka_iu_i\) exists.
We divide the proof into four steps.  Steps~1--3 are elementary and
show that \(M_a=\inf_{w>0}K_a(w)\) for an explicit function \(K_a\).
Step~4 identifies this infimum with the right spectral edge.
 
\medskip
\noindent
\textbf{Step 1: Calculation of inverse Cauchy-transform branch.}
 
By definition, \(K_{\mu_{a,t}}(w)=1/w+R_{\mu_{a,t}}(w)\) for \(w\)
sufficiently close to \(0\).  By \eqref{eq:app-R-mu-a-t} and the
identity \(a_iR_t(a_iw/k)=(k/w)\,g_t(a_iw/k)\), it follows that
\(K_{\mu_{a,t}}(w)=K_a(w)\) for small \(w>0\), where
\begin{equation}
    K_a(w)
    :=
    \frac1w
    \left(
      1+k\sum_{i=1}^kg_t\left(\frac{a_iw}{k}\right)
    \right),
    \qquad
    w>0.
    \label{eq:app-K-a}
\end{equation}
The right-hand side is real analytic on \((0,\infty)\), because
\(\Delta_t>0\) on \(\mathbb R\).  Thus, for small \(w>0\), \(K_a(w)\) is
the actual inverse of \(G_{\mu_{a,t}}\); for larger \(w\), it is only
the real-analytic continuation of this branch.
 
For \(w\geq0\), put \(y_i(w):=a_iw/k\) and \(u_i(w):=u_t(y_i(w))\), let
\(u(w)\) be the vector with entries \(u_i(w)\), and put
\(F_a(w):=\sum_{i=1}^kc_t(u_i(w))\).  Then \(u(w)\in(0,1)^k\), the
function \(F_a\) is continuous and nonnegative on \([0,\infty)\), and
\(F_a(0)=kc_t(t)=0\).  By the equality case of
\eqref{eq:app-legendre-inequality},
\(g_t(y_i(w))=y_i(w)u_i(w)-c_t(u_i(w))\).  Multiplying by \(k\), summing
over \(i\), and using \(ky_i(w)=a_iw\), we obtain
\(k\sum_{i=1}^kg_t(y_i(w))=w\sum_{i=1}^ka_iu_i(w)-kF_a(w)\).
Substituting this into \eqref{eq:app-K-a} gives the tangency identity
\begin{equation}
    K_a(w)
    =
    \sum_{i=1}^ka_iu_i(w)
    +
    \frac{1-kF_a(w)}{w},
    \qquad
    w>0.
    \label{eq:app-K-a-tangency}
\end{equation}
Differentiating \eqref{eq:app-K-a} and using \(g_t'=u_t\), we find
\(K_a'(w)=w^{-1}\bigl(\sum_{i=1}^ka_iu_i(w)-K_a(w)\bigr)\).  Hence, by
\eqref{eq:app-K-a-tangency},
\begin{equation}
    K_a'(w)
    =
    \frac{kF_a(w)-1}{w^2},
    \qquad
    w>0.
    \label{eq:app-K-a-derivative}
\end{equation}
 
\medskip
\noindent
\textbf{Step 2: An inequality for \eqref{eq:app-edge-support-duality}.}
 
Fix \(u\in\mathscr D_{k,t}\) and \(w>0\).  Apply
\eqref{eq:app-legendre-inequality} with \(y=a_iw/k\), multiply by
\(k/w\), and sum over \(i\).  By \eqref{eq:app-K-a}, this gives
\(\sum_{i=1}^ka_iu_i\leq K_a(w)-1/w+(k/w)\sum_{i=1}^kc_t(u_i)\).  Since
\(\sum_{i=1}^kc_t(u_i)\leq1/k\), we obtain
\(\sum_{i=1}^ka_iu_i\leq K_a(w)\).  As \(u\) and \(w\) were arbitrary,
\begin{equation}
    M_a
    \leq
    \inf_{w>0}K_a(w).
    \label{eq:app-support-upper}
\end{equation}
In particular, \(\inf_{w>0}K_a(w)\) is finite.
 
\medskip
\noindent
\textbf{Step 3: Attainment of the inequality.}
 
Let \(w_0>0\) satisfy \(kF_a(w_0)=1\).  Then
\(u(w_0)\in\mathscr D_{k,t}\), and \eqref{eq:app-K-a-tangency} gives
\(K_a(w_0)=\sum_{i=1}^ka_iu_i(w_0)\leq M_a\).  Since also
\(K_a(w_0)\geq\inf_{w>0}K_a(w)\geq M_a\) by
\eqref{eq:app-support-upper}, we conclude that
\begin{equation}
    K_a(w_0)
    =
    \min_{w>0}K_a(w)
    =
    M_a.
    \label{eq:app-tangent-min}
\end{equation}
 
Suppose first that \(kF_a(w_1)\geq1\) for some \(w_1>0\).  Since
\(kF_a\) is continuous on \([0,w_1]\) and \(kF_a(0)=0\), the
intermediate value theorem gives \(w_0\in(0,w_1]\) with
\(kF_a(w_0)=1\).  Thus \(M_a=\inf_{w>0}K_a(w)\) by
\eqref{eq:app-tangent-min}.
 
Suppose instead that \(kF_a(w)<1\) for every \(w>0\).  Then
\(u(w)\in\mathscr D_{k,t}\) for every \(w>0\).  Since \(F_a\geq0\),
\eqref{eq:app-K-a-tangency} gives
\(M_a\geq\sum_{i=1}^ka_iu_i(w)\geq K_a(w)-1/w\geq\inf_{w'>0}K_a(w')-1/w\).
Letting \(w\to\infty\), we again obtain \(M_a\geq\inf_{w>0}K_a(w)\).
Combining both cases with \eqref{eq:app-support-upper} gives
\begin{equation}
    M_a
    =
    \inf_{w>0}K_a(w).
    \label{eq:app-support-inf}
\end{equation}
 
\medskip
\noindent
\textbf{Step 4: the infimum of \(K_a\) is the right spectral edge.}
 
Write \(\mu:=\mu_{a,t}\), \(G:=G_\mu\), \(r_a:=r_\mu\), and
\(w_*:=\omega_\mu\in(0,\infty]\).  By
Appendix~\ref{app:sub-basics}, \(G\) maps \((r_a,\infty)\) bijectively
onto \((0,w_*)\).  Its real inverse \(\kappa_a\), defined on
\((0,w_*)\), agrees with \(K_\mu\), and hence with \(K_a\), for small
\(w>0\).  Since \(G\) is real analytic with \(G'<0\) on
\((r_a,\infty)\), the function \(\kappa_a\) is real analytic with
\(\kappa_a'<0\).  The identity theorem therefore gives \(\kappa_a=K_a\)
on \((0,w_*)\).  Consequently, \(K_a'<0\) on \((0,w_*)\), and
\(r_a=\lim_{w\uparrow w_*}K_a(w)\).
 
If \(w_*=\infty\), then \(r_a=\inf_{w>0}K_a(w)\), because \(K_a\) is
decreasing on all of \((0,\infty)\).
 
Suppose now that \(w_*<\infty\).  Then \(r_a=K_a(w_*)\) by continuity,
and \(K_a'(w_*)\leq0\).  We claim that in fact \(K_a'(w_*)=0\).  Assume,
to the contrary, that \(K_a'(w_*)<0\).  Since \(K_a\) is real analytic
at \(w_*\), the analytic inverse-function theorem gives an open
neighborhood \(W\) of \(w_*\) and an open disc \(V\) centered at \(r_a\)
such that the holomorphic extension of \(K_a\) maps \(W\)
biholomorphically onto \(V\).  Let \(\psi\colon V\to W\) be the inverse.
After shrinking \(V\), the real inverse-function theorem shows that
every \(x\in V\cap\mathbb R\) has a real preimage in \(W\), so \(\psi\) is
real on \(V\cap\mathbb R\).  For real \(x>r_a\) close to \(r_a\), we have
\(G(x)\in W\cap(0,w_*)\) and \(K_a(G(x))=\kappa_a(G(x))=x\), so
\(G(x)=\psi(x)\).  Since \(\operatorname{supp}(\mu)\subset(-\infty,r_a]\),
both \(G\) and \(\psi\) are holomorphic on the connected set
\(V\setminus(-\infty,r_a]\), and the identity theorem gives \(G=\psi\)
there.
 
Now choose \(\delta,\eta_0>0\) such that \(r_a\pm\delta\) are not atoms
of \(\mu\) and \([r_a-\delta,r_a+\delta]\times[0,\eta_0]\subset V\).
Since \(\psi\) is uniformly continuous on this rectangle and real on its
bottom edge, \(\operatorname{Im}G(x+\mathrm i\eta)
=\operatorname{Im}\psi(x+\mathrm i\eta)\to0\) as \(\eta\downarrow0\),
uniformly in \(x\in[r_a-\delta,r_a+\delta]\).  For \(\alpha<\beta\) that
are not atoms, the Stieltjes inversion formula gives
\(\mu((\alpha,\beta))=-\pi^{-1}\lim_{\eta\downarrow0}
\int_\alpha^\beta\operatorname{Im}G(x+\mathrm i\eta)\,dx\).  Hence
\(\mu((r_a-\delta,r_a+\delta))=0\).  This contradicts
\(r_a\in\operatorname{supp}(\mu)\).
 
Therefore \(K_a'(w_*)=0\), that is, \(kF_a(w_*)=1\) by
\eqref{eq:app-K-a-derivative}.  By \eqref{eq:app-tangent-min},
\(r_a=K_a(w_*)=\min_{w>0}K_a(w)\).  In both cases, we have proved
\begin{equation}
    \max\operatorname{supp}(\mu_{a,t})
    =
    \inf_{w>0}K_a(w).
    \label{eq:app-right-edge}
\end{equation}
Combining this with \eqref{eq:app-support-inf} proves
\eqref{eq:app-edge-support-duality}.
\end{proof}


\subsection{Second moments of the normalized Choi matrix}
This subsection aims to prove the asymptotic second moment of the normalized Choi matrix used in the proof of Theorem~\ref{thm:bell-output-limit}. Put \(D_{k,t}:=k^4t-2k^2t+1\) in this subsection only.
We use the full strong block-modification theorem~\cite{Nechita2018}, including convergence
of normalized traces. For \(X>0\), with eigenvalues \(a_1,\ldots,a_k\),
the limiting law of
\(T_n(X)=\Tr_B[P_n(I_n\otimes X)]\) is
\begin{equation}
    \mu_X=\mathop{\boxplus}_{i=1}^k
             \bigl(D_{a_i/k}b_t\bigr)^{\boxplus k}.
    \label{eq:bell-scalar-law}
\end{equation}
The random compression formula gives a positive lower spectral bound
for \(P_{A,n}=T_n(I_k)\), since \(t>k^{-2}\).
Consequently, for each fixed \(X>0\), almost surely,
\begin{equation}
    F_n(X):=\frac1n\log\det T_n(X)
        \longrightarrow F(X):=\int\log x\,d\mu_X(x).
    \label{eq:bell-log-limit}
\end{equation}
Indeed, strong convergence places the spectra in a fixed compact
subset of \((0,\infty)\), where the logarithm is continuous.

\begin{lemma}[Second logarithmic variation]
\label{lem:bell-log-hessian}
For every Hermitian \(H\), the function
\(\varepsilon\mapsto F(I_k+\varepsilon H)\) has second derivative
\begin{equation}
    \left.\frac{d^2}{d\varepsilon^2}F(I_k+\varepsilon H)\right|_0
      =c_0(\Tr H)^2+c_1\Tr(H^2),
    \qquad
    c_1=-\frac{k(1-t)}{D_{k,t}},\quad
    c_0=\frac{-1-kc_1}{k^2}.
    \label{eq:app-invariant-hessian}
\end{equation}
\end{lemma}

\begin{proof}
Write \(R_t\) for the Bernoulli \(R\)-transform,
\(g_t(y)=yR_t(y)\), and \(L_t(y)=\int_0^yR_t(v)\,dv\).
If \(w_X=G_{\mu_X}(0)<0\), the inverse-Cauchy identity gives
\begin{equation}
    1+k\sum_i g_t(a_iw_X/k)=0,\qquad
    F(X)=-\log(-w_X)-1-k\sum_i L_t(a_iw_X/k).
    \label{eq:bell-log-potential}
\end{equation}
For completeness, the second identity follows by differentiating
\(-\log(-G(z))+zG(z)-1-k\sum_iL_t(a_iG(z)/k)\) in \(z<0\).
Its derivative is \(G(z)\), and its difference from \(\log(-z)\)
tends to zero as \(z\to-\infty\), just as for
\(\int\log(x-z)\,d\mu_X(x)\).
The inverse identity continues from its germ at infinity to this
negative branch. Since \(g_t'>0\), its root is simple and depends
smoothly on \(a\).

Let \(w_0=w_{I_k}\) and \(y_0=w_0/k\). Then
\begin{equation}
    g_t(y_0)=-k^{-2},\qquad
    y_0=-\frac{k^2-1}{k^2(k^2t-1)}.
    \label{eq:bell-negative-root}
\end{equation}
For a traceless \(H\), let \(a_i(\varepsilon)=1+\varepsilon h_i\),
where \(h_i\) are its eigenvalues. The first variation of the root
equation vanishes because \(\sum_i h_i=0\); hence
\(w'(0)=0\). Moreover, the derivative of the right side of
\eqref{eq:bell-log-potential} with respect to \(w_X\) is zero at
the root. Differentiating therefore gives
\begin{equation}
    \left.\frac d{d\varepsilon}F(I_k+\varepsilon H)\right|_0=0,
    \qquad
    \left.\frac{d^2}{d\varepsilon^2}F(I_k+\varepsilon H)\right|_0
      =-k y_0^2R_t'(y_0)\Tr(H^2).
    \label{eq:bell-traceless-hessian}
\end{equation}
The Bernoulli identity and its derivative yield
\begin{equation}
    g_t(y)^2+(1-y)g_t(y)-ty=0,\qquad
    y^2R_t'(y)=
      \frac{g_t(y)^2(1-t)}{g_t(y)^2+2tg_t(y)+t}.
\end{equation}
Substituting \eqref{eq:bell-negative-root} gives
\(c_1=-k(1-t)/D_{k,t}\). Finally,
\(F(xX)=\log x+F(X)\) for \(x>0\). Splitting \(H\) into its scalar and
traceless parts gives
\eqref{eq:app-invariant-hessian}.
\end{proof}

\begin{proposition}[Second moments after local normalization]
\label{prop:normalized-choi-block-limit}
Under the assumptions of Theorem~\ref{thm:bell-output-limit}, almost
surely, for all \(i,j,p,q\),
\begin{equation}
    \frac1n\Tr(J_{ij}^{(n)}J_{qp}^{(n)})
      \longrightarrow
      -c_0\delta_{ij}\delta_{pq}-c_1\delta_{ip}\delta_{jq},
    \label{eq:normalized-block-second-moment}
\end{equation}
where \(c_0,c_1\) are given in
\eqref{eq:app-invariant-hessian}.
\end{proposition}

\begin{proof}
For Hermitian \(H\), put
\begin{equation}
    C_n(H):=P_{A,n}^{-1/2}T_n(H)P_{A,n}^{-1/2}.
\end{equation}
This is the unital positive adjoint map of the normalized channel,
up to transposition. Thus \(\|C_n(H)\|_\infty\leq\|H\|_\infty\).
The determinant identity
\begin{equation}
    F_n(I_k+xH)-F_n(I_k)
       =\frac1n\log\det(I_n+xC_n(H))
\end{equation}
and the scalar Taylor expansion give, for
\(0<|x|\) and \(|x|\|H\|_\infty\leq\tfrac12\),
\begin{equation}
 \left|
  \frac{F_n(I_k+xH)+F_n(I_k-xH)-2F_n(I_k)}{x^2}
  +\frac1n\Tr(C_n(H)^2)
 \right|
 \leq4x^2\|H\|_\infty^4.
 \label{eq:bell-uniform-finite-difference}
\end{equation}
Taking \(n\to\infty\) and then \(x\to0\) proves
\begin{equation}
    \frac1n\Tr(C_n(H)^2)
       \longrightarrow-c_0(\Tr H)^2-c_1\Tr(H^2).
\end{equation}
Real polarization, followed by complex linearity, gives the corresponding
mixed moments. Since \(C_n(E_{ji})=J_{ij}^{(n)}\),
\eqref{eq:normalized-block-second-moment} follows. Only finitely many
Hermitian directions and countably many values of \(x\) are needed,
so these limits hold on one event of probability one.
\end{proof}

\begin{proposition}[Asymptotic second moments]
\label{prop:asymptotic-choi-purity}
Under the same assumptions, almost surely,
\begin{equation}
    \frac1{nk}\Tr(J_n^2)
       \longrightarrow
       \frac{k^4-(1+t)k^2+1}{k^2(k^4t-2k^2t+1)}.
    \label{eq:app-normalized-choi-purity}
\end{equation}
\end{proposition}

\begin{proof}
Summing \eqref{eq:normalized-block-second-moment} with \(p=i\)
and \(q=j\), and dividing by \(k\), gives the limit
\(-c_0-kc_1=\alpha_{k,t}\).
Equivalently, the exact identity
\begin{equation}
    \langle\psi_k^+|Z_n|\psi_k^+\rangle
        =\frac1{nk}\Tr(J_n^2)
    \label{eq:bell-overlap-purity}
\end{equation}
follows by summing \eqref{eq:Z-block-entry}.
\end{proof}


\section{Entropy asymptotics}\label{app:entropy-asymptotics}
This appendix section proves the large-\(k\) entropy estimates of the main
text: the single-output estimates of Corollary~\ref{cor:single-asymptotic}, the Bell-output
estimates of Corollary~\ref{cor:joint-asymptotic}; the single-output and Bell-output estimates after antisymmetric postprocessing in Theorem~\ref{thm:entropy-asymptotics-antisymmetric}. We record only the estimates in Theorem~\ref{thm:entropy-asymptotics-antisymmetric}, since
\(\mc E_{k,1}=\mathrm{id}\), the case \(r=1\) gives the estimates in Corollary~\ref{cor:single-asymptotic} and Corollary~\ref{cor:joint-asymptotic}

Throughout, \(t\in(0,1)\),
\(p\in(0,\infty)\), and an integer \(r\geq1\) are fixed, \(\gamma=(1-t)/t\), and
\(k\to\infty\). We write \(D=\binom kr\) and \(N=\binom{k-1}{r-1}\), so that \(kN=rD\) by
\eqref{eq:DN}.

\begin{proposition}[Entropy of the one-channel output body]
\label{prop:app-antisymmetric-single-entropy-asymptotics}
For every fixed \(t\in(0,1)\), \(p\in(0,\infty)\), and integer \(r\geq1\),
\begin{equation}
    \min_{\sigma\in\mathscr K_{k,t}}S_p\bigl(\mc E_{k,r}(\sigma)\bigr)
    =\log\binom kr-\frac{2p(1-t)}{trk^2}+O_{p,r,t}(k^{-5/2}).
    \label{eq:app-antisymmetric-single-entropy-asymptotics}
\end{equation}
In particular, for \(r=1\),
\[
    \min_{\sigma\in\mathscr K_{k,t}}S_p(\sigma)=\log k-\frac{2p(1-t)}{tk^2}+O_{p,t}(k^{-5/2}).
\]
\end{proposition}

\begin{proof}
By Proposition~\ref{prop:eigenvalue-shuffling}, the spectrum of \(\mc E_{k,r}(\sigma)\)
depends only on the spectrum of \(\sigma\), so it is enough to minimize over the
eigenvalue body \(\Lambda_{k,t}\). For all sufficiently large \(k\), its elements are the
normalized vectors
\[
    \lambda_i=\frac{u_i}{\sum_ju_j},\qquad u\in\mathscr D_{k,t}.
\]
Write \(x_i=u_i-t\), \(\bar x=k^{-1}\sum_ix_i\), and
\[
    q_i=\frac{x_i-\bar x}{t+\bar x}.
\]
Then \(\lambda_i=k^{-1}(1+q_i)\) and \(\sum_iq_i=0\).

Since \(c_t\geq0\) has the unique zero \(t\) and \(c_t(u_i)\leq k^{-1}\), compactness first
gives \(\max_i|u_i-t|\to0\). Quadratic comparability near \(t\), followed by the expansion
\[
    c_t(t+x)=\frac{x^2}{4t(1-t)}+O_t(|x|^3)
\]
and the constraint \(\sum_ic_t(u_i)\leq k^{-1}\), then imply, uniformly in
\(u\in\mathscr D_{k,t}\),
\[
    \max_i|x_i|=O_t(k^{-1/2}),\qquad
    \sum_ix_i^2\leq\frac{4t(1-t)}k+O_t(k^{-3/2}),\qquad
    \bar x=O_t(k^{-1}).
\]
Consequently,
\begin{equation}
    \max_i|q_i|=O_t(k^{-1/2}),\qquad
    \sum_iq_i^2\leq\frac{4(1-t)}{tk}+O_t(k^{-3/2}).
    \label{eq:app-q-bounds}
\end{equation}

The leading constant is attained. Indeed, for \(v_i=\sqrt2\cos(2\pi i/k)\), set
\[
    u_i=t+\frac{2\sqrt{t(1-t)}}k\bigl(1-k^{-1/2}\bigr)v_i.
\]
Since \(\sum_iv_i=0\), \(\sum_iv_i^2=k\), and \(\max_i|v_i|\leq\sqrt2\), the preceding
Taylor expansion shows that \(u\in\mathscr D_{k,t}\) for all sufficiently large \(k\), and
its associated vector satisfies
\[
    \sum_iq_i^2=\frac{4(1-t)}{tk}+O_t(k^{-3/2}).
\]
Thus the supremum of \(\sum_iq_i^2\) over \(\mathscr D_{k,t}\) equals
\(4(1-t)/(tk)+O_t(k^{-3/2})\).

If \(\sigma\) has spectrum \(\lambda\), Proposition~\ref{prop:eigenvalue-shuffling} and
\(kN=rD\) give the eigenvalues
\[
    \mu_I=\frac1N\sum_{i\in I}\lambda_i=\frac{1+\delta_I}D,\qquad
    \delta_I:=\frac1r\sum_{i\in I}q_i,
\]
of \(\mc E_{k,r}(\sigma)\), indexed by the \(r\)-subsets \(I\subseteq[k]\). Thus
\(\sum_I\delta_I=0\) and \(\max_I|\delta_I|\leq\max_i|q_i|\). Since \(\sum_iq_i=0\), we
have \(\sum_{i\in I}q_i=-\sum_{j\notin I}q_j\), and for \(i\neq j\) exactly
\(\binom{k-2}{r-1}\) of the \(r\)-subsets contain \(i\) but not \(j\). Hence
\[
    \frac1D\sum_I\delta_I^2
    =-\frac1{r^2D}\sum_I\sum_{i\in I}\sum_{j\notin I}q_iq_j
    =-\frac{\binom{k-2}{r-1}}{r^2D}\sum_{i\neq j}q_iq_j
    =\frac{k-r}{rk(k-1)}\sum_iq_i^2.
\]
For \(p\neq1\), a uniform Taylor expansion gives
\[
    S_p(\mu)=\log D-\frac p{2D}\sum_I\delta_I^2+O_{p,r,t}(k^{-5/2})
    =\log D-\frac p{2rk}\sum_iq_i^2+O_{p,r,t}(k^{-5/2}).
\]
The same formula holds for \(p=1\), using \((1+q)\log(1+q)=q+q^2/2+O(|q|^3)\). Maximizing
the quadratic term and using \eqref{eq:app-q-bounds} proves
\eqref{eq:app-antisymmetric-single-entropy-asymptotics}. For \(r=1\), we have
\(\mc E_{k,1}=\mathrm{id}\) and \(D=k\).
\end{proof}

\begin{proposition}[Entropy of the limiting Bell output]
\label{prop:app-antisymmetric-bell-entropy-asymptotics}
Let \(Z_{k,t}=r_{k,t}\psi_k^++(1-r_{k,t})I_{k^2}/k^2\) be the limit in
Theorem~\ref{thm:bell-output-limit}. For every fixed \(t\in(0,1)\), \(p\in(0,\infty)\), and
integer \(r\geq1\),
\begin{equation}
    S_p\Bigl(\bigl(\mc E_{k,r}\otimes\overline{\mc E}_{k,r}\bigr)(Z_{k,t})\Bigr)
    =2\log\binom kr-\frac{B_{p,r}(\gamma)}{k^2}+O_{p,r,t}(k^{-3}),
    \label{eq:app-antisymmetric-bell-entropy-asymptotics}
\end{equation}
where \(B_{p,r}(\gamma)\) is defined in \eqref{eq:Bpr} for \(p\neq1\), and
\(B_{1,r}(\gamma)=(r^2+\gamma)\log(1+\gamma/r^2)-\gamma\) is its continuous extension. The case $r = 1$ proves
Corollary~\ref{cor:joint-asymptotic}.
\end{proposition}

\begin{proof}
Put \(m=k^2\), \(s=t^{-1}\), and \(\delta=s-1=\gamma\). Theorem~\ref{thm:bell-output-limit}
gives
\[
    r_{k,t}=\frac\delta m+O_t(m^{-2}).
\]

We first treat \(r=1\). Then the state is \(Z_{k,t}\), whose spectrum is
\(\lambda^{\rm Bell}_{k,t}\). If \(\alpha_{k,t}\) is the exceptional Bell eigenvalue and
\(\beta_{k,t}\) is the eigenvalue of multiplicity \(m-1\), then
\[
    x_m:=m\alpha_{k,t}=1+(m-1)r_{k,t}=s+O_t(m^{-1}),\qquad
    y_m:=m\beta_{k,t}=1-r_{k,t}=1-\frac\delta m+O_t(m^{-2}).
\]
For fixed \(p\neq1\), it follows that
\[
    x_m^p+(m-1)y_m^p=m+s^p-1-p(s-1)+O_{p,t}(m^{-1}).
\]
Substitution into the definition of \(S_p\), followed by \(\log(1+z)=z+O(z^2)\), yields
\[
    S_p\bigl(\lambda^{\rm Bell}_{k,t}\bigr)
    =\log m-\frac{s^p-1-p(s-1)}{(p-1)m}+O_{p,t}(m^{-2}).
\]
For \(p=1\), the identity \(x_m+(m-1)y_m=m\) similarly gives
\[
    S_1\bigl(\lambda^{\rm Bell}_{k,t}\bigr)=\log m-\frac{s\log s-s+1}m+O_t(m^{-2}).
\]
These coefficients are \(A_p(t)=B_{p,1}(\gamma)\) and \(A_1(t)=B_{1,1}(\gamma)\),
respectively; setting \(m=k^2\) proves the case \(r=1\), with error \(O_{p,t}(k^{-4})\).

Now let \(r\geq2\) and \(k\geq2r\). We first fix the notation for the antisymmetric
calculation. For \(0\leq\ell\leq r\), let \(\Pi_\ell\) be the orthogonal projection onto
the antisymmetric subspace \(\mc A_\ell\subseteq(\mb C^k)^{\otimes\ell}\), with
\(\Pi_0=1\) and \(\Pi_r=\Pi^{\asym}_{k,r}\). We identify \(\mc B(\mc A_\ell)\) with the
operators \(Y=\Pi_\ell Y\Pi_\ell\) on \((\mb C^k)^{\otimes\ell}\), equipped with the
Hilbert--Schmidt inner product. Let \(F_{jl}\) be the swap of the \(j\)th and \(l\)th tensor
factors, so that \(\Pi_\ell F_{jl}=F_{jl}\Pi_\ell=-\Pi_\ell\) for \(j\neq l\). For
\(X\in\mb M_k\), write \(X^{(j)}=I_k^{\otimes(j-1)}\otimes X\otimes I_k^{\otimes(\ell-j)}\);
let \(E_{ab}=|a\rangle\langle b|\) be the matrix units, and put
\[
    L^{(\ell)}_{ab}:=\Pi_\ell\Bigl(\sum_{j=1}^\ell E^{(j)}_{ab}\Bigr)\Pi_\ell,
    \qquad a,b\in[k].
\]
Since \(X^{(j)}=F_{1j}X^{(1)}F_{1j}\), we have \(\Pi_rX^{(j)}\Pi_r=\Pi_rX^{(1)}\Pi_r\).
Hence, by \(kN=rD\), the definition of \(\mc E_{k,r}\) reads
\[
    \mc E_{k,r}(X)
    =\frac kD\,\Pi_r\bigl(X\otimes I_k^{\otimes(r-1)}\bigr)\Pi_r
    =\frac1N\,\Pi_r\Bigl(\sum_{j=1}^rX^{(j)}\Bigr)\Pi_r .
\]
In particular, \(\mc E_{k,r}(E_{ab})=N^{-1}L^{(r)}_{ab}\) and
\(\mc E_{k,r}(I_k/k)=\Pi_r/D\). Since \(\Pi_r\) is a real matrix,
\(\overline{\mc E}_{k,r}=\mc E_{k,r}\). As \(\psi_k^+=k^{-1}\sum_{a,b}E_{ab}\otimes E_{ab}\),
the Bell output of \(\mc E_{k,r}\otimes\overline{\mc E}_{k,r}\) is
\[
    W_{k,r}:=\bigl(\mc E_{k,r}\otimes\overline{\mc E}_{k,r}\bigr)(\psi_k^+)
    =\frac1{kN^2}\sum_{a,b=1}^kL^{(r)}_{ab}\otimes L^{(r)}_{ab},
\]
and
\[
    \bigl(\mc E_{k,r}\otimes\overline{\mc E}_{k,r}\bigr)(Z_{k,t})
    =r_{k,t}W_{k,r}+(1-r_{k,t})\,\frac{\Pi_r\otimes\Pi_r}{D^2}.
\]
Recall the Slater vectors
\[
    |e_I\rangle=\frac1{\sqrt{r!}}\sum_{\pi\in\mathfrak S_r}\operatorname{sgn}(\pi)
    |i_{\pi(1)}\cdots i_{\pi(r)}\rangle,\qquad I=\{i_1<\cdots<i_r\}\subseteq[k],
\]
form a real orthonormal basis of \(\mc A_r\). In this basis, the matrices of the \(L^{(r)}_{ab}\) are real and
\((L^{(r)}_{ab})^*=L^{(r)}_{ba}\), so \((L^{(r)}_{ab})^T=L^{(r)}_{ba}\). The unitary
vectorization \(\mathrm{vec}(|e_I\rangle\langle e_J|)=|e_I\rangle\otimes|e_J\rangle\)
satisfies \((X_1\otimes X_2)\,\mathrm{vec}(Y)=\mathrm{vec}(X_1YX_2^T)\). Therefore
\(W_{k,r}\) is unitarily equivalent to the map
\(Y\mapsto(kN^2)^{-1}\sum_{a,b}L^{(r)}_{ab}YL^{(r)}_{ba}\) on \(\mc B(\mc A_r)\).

We now compute the spectrum of this map. For \(1\leq\ell\leq r\), let
\(R_\ell(Y):=\Tr_\ell Y\) trace out the last factor. It maps \(\mc B(\mc A_\ell)\) into
\(\mc B(\mc A_{\ell-1})\), because \(\Pi_\ell=(\Pi_{\ell-1}\otimes I_k)\Pi_\ell\), and its
Hilbert--Schmidt adjoint is \(R_\ell^*(Z)=\Pi_\ell(Z\otimes I_k)\Pi_\ell\). For
\(Y\in\mc B(\mc A_\ell)\), we expand \(\sum_{a,b}L^{(\ell)}_{ab}YL^{(\ell)}_{ba}\) using
\(\sum_{a,b}E_{ab}XE_{ba}=\Tr(X)I_k\), \(E^{(l)}_{ba}=F_{jl}E^{(j)}_{ba}F_{jl}\), and
\(YF_{jl}=-Y=F_{jl}Y\). After compression by \(\Pi_\ell\), each of the \(\ell^2\) terms
indexed by pairs \((j,l)\) of tensor factors equals
\(\Pi_\ell\bigl((\Tr_\ell Y)\otimes I_k\bigr)\Pi_\ell\). Thus
\[
    \sum_{a,b}L^{(\ell)}_{ab}YL^{(\ell)}_{ba}=\ell^2R_\ell^*R_\ell(Y).
\]
The coset decomposition of \(\mathfrak S_\ell\) over \(\mathfrak S_{\ell-1}\) gives
\(\Pi_\ell=\ell^{-1}\bigl(I-\sum_{l<\ell}F_{l\ell}\bigr)(\Pi_{\ell-1}\otimes I_k)\). Insert
this and its adjoint on both sides of \(Z\otimes I_k\), and trace out the last factor. The
identity and the \(2(\ell-1)\) terms containing one swap contribute \((k-2\ell+2)Z\), and
each of the \((\ell-1)^2\) terms containing two swaps contributes
\(R_{\ell-1}^*R_{\ell-1}(Z)\). Hence
\[
    R_\ell R_\ell^*
    =\frac1{\ell^2}\Bigl((k-2\ell+2)\,\mathrm{id}+(\ell-1)^2R_{\ell-1}^*R_{\ell-1}\Bigr)
    \qquad\text{on }\mc B(\mc A_{\ell-1}).
\]
Put
\[
    d_j:=\binom kj^2-\binom k{j-1}^2,\qquad j=0,\ldots,r.
\]
Since \(R_\ell^*R_\ell\) and \(R_\ell R_\ell^*\) have the same nonzero eigenvalues, we
argue by induction on \(\ell\), using
\((\ell-1-j)(k-\ell-j+2)+k-2\ell+2=(\ell-j)(k-\ell-j+1)\) and
\(\dim\mc B(\mc A_\ell)-\dim\mc B(\mc A_{\ell-1})=d_\ell\). This shows that \(R_r^*R_r\)
has the eigenvalues \((r-j)(k-r-j+1)/r^2\) with multiplicities \(d_j\). Consequently,
\(W_{k,r}\) has the eigenvalues \(w_j=(r-j)(k-r-j+1)/(kN^2)\) with multiplicities \(d_j\),
and \((\mc E_{k,r}\otimes\overline{\mc E}_{k,r})(Z_{k,t})\) has the eigenvalues
\(\nu_j=r_{k,t}w_j+(1-r_{k,t})/D^2\) with multiplicities \(d_j\).

Put \(x_j:=D^2\nu_j=1+r_{k,t}(D^2w_j-1)\). Then \(\sum_jd_j=\sum_jd_jx_j=D^2\), and
\(kN=rD\) gives \(D^2w_j=k(r-j)(k-r-j+1)/r^2\). Hence \(x_r=1-r_{k,t}\),
\(x_{r-1}=1+\gamma/r^2+O_{r,t}(k^{-1})\), and \(x_j=O_{r,t}(1)\) for all \(j\). Moreover,
\(\binom kj/\binom k{j+1}=(j+1)/(k-j)\) gives \(d_{r-1}/D^2=r^2/k^2+O_r(k^{-3})\) and
\(d_j/D^2=O_r(k^{-4})\) for \(j\leq r-2\). For fixed \(p\neq1\), put
\(g_p(x)=x^p-1-p(x-1)\). Then
\[
    \sum_jd_jx_j^p=D^2+\sum_jd_jg_p(x_j)
    =D^2\Bigl(1+\frac{r^2}{k^2}\,g_p\bigl(1+\gamma/r^2\bigr)+O_{p,r,t}(k^{-3})\Bigr).
\]
Indeed, the term \(j=r\) contributes \(O(D^2r_{k,t}^2)\), the term \(j=r-1\) the
displayed main term, and the terms \(j\leq r-2\) contribute \(O(D^2k^{-4})\).
Substitution into the definition of \(S_p\), followed by \(\log(1+z)=z+O(z^2)\), yields
\[
    S_p\Bigl(\bigl(\mc E_{k,r}\otimes\overline{\mc E}_{k,r}\bigr)(Z_{k,t})\Bigr)
    =2\log D+\frac{r^2g_p(1+\gamma/r^2)}{(1-p)k^2}+O_{p,r,t}(k^{-3}),
\]
and \(r^2g_p(1+\gamma/r^2)/(1-p)=-B_{p,r}(\gamma)\). For \(p=1\), the identities
\(\sum_jd_jx_j=\sum_jd_j=D^2\) similarly give
\[
    S_1\Bigl(\bigl(\mc E_{k,r}\otimes\overline{\mc E}_{k,r}\bigr)(Z_{k,t})\Bigr)
    =2\log D-\frac{r^2h(1+\gamma/r^2)}{k^2}+O_{r,t}(k^{-3}),\qquad
    h(x):=x\log x-x+1,
\]
and \(r^2h(1+\gamma/r^2)=B_{1,r}(\gamma)\). This proves
\eqref{eq:app-antisymmetric-bell-entropy-asymptotics}.
\end{proof}

\section{Certification of the output dimension}
\label{app:k182-certificate}
 
This appendix verifies the one-copy lower bound and records the interval certificate used
in the proof of Proposition~\ref{prop:certified-k182}. Throughout, \(k\geq2\), and \(c_t\),
\(\mathscr D_{k,t}\), and \(\Lambda_{k,t}\) are as in \eqref{eq:set-D} and
\eqref{eq:def-Lambda-kt}. Writing \(t=\sin^2\theta\) and \(u=\sin^2\varphi\) with
\(\theta,\varphi\in[0,\pi/2]\), we have
\begin{equation}
 c_t(u)=\bigl(\sin\theta\cos\varphi-\cos\theta\sin\varphi\bigr)^2=\sin^2(\varphi-\theta).
 \label{eq:ct-angle}
\end{equation}
 
\subsection{The one-copy minimizer}
 
\begin{lemma}[Reduction of the entropy minimum]
\label{lem:finite-k-entropy-reduction}
Assume
\[
 \frac1{k(k-1)}<t<1-\frac1k,\qquad
 \bar u_{k,t}:=\Bigl(\sqrt{t(1-k^{-1})}+\sqrt{(1-t)k^{-1}}\Bigr)^2<\frac14.
\]
Then every minimizer of \(S_1\) on \(\Lambda_{k,t}\) is, up to a permutation of the
coordinates, of the form
\[
 \Bigl(x,\frac{1-x}{k-1},\ldots,\frac{1-x}{k-1}\Bigr),\qquad x>\frac1k.
\]
Consequently, with \(L^*_{k,t}:=\max_{\lambda\in\Lambda_{k,t}}\lambda_1\) and
\(h_k(x):=-x\log x-(1-x)\log\frac{1-x}{k-1}\),
\begin{equation}
 \min_{\sigma\in\mathscr K_{k,t}}S_1(\sigma)=h_k(L^*_{k,t})\geq h_k(L)
 \qquad\text{for every }L\in[L^*_{k,t},1).
 \label{eq:finite-k-entropy-bound}
\end{equation}
\end{lemma}
 
\begin{proof}
Put \(C(u)=\sum_ic_t(u_i)\), \(s(u)=\sum_iu_i\), and \(F(u)=S_1(u/s(u))\), so that
\(\min_{\mathscr K_{k,t}}S_1=\min_{\mathscr D_{k,t}}F\). A minimizer \(u\) exists by
compactness. We divide the proof into four steps.
 
\medskip\noindent\textbf{Step 1.}
By \eqref{eq:ct-angle}, \(c_t(v)\leq1/k\) forces \(|\varphi-\theta|\leq\arcsin k^{-1/2}\),
and hence \(v\leq\sin^2(\theta+\arcsin k^{-1/2})=\bar u_{k,t}\); here \(t<1-1/k\) gives
\(\theta+\arcsin k^{-1/2}<\pi/2\). Hence every point of \(\mathscr D_{k,t}\) has
coordinates in \([0,\bar u_{k,t}]\subseteq[0,1/4)\). A vector with at most one positive
coordinate is not in \(\mathscr D_{k,t}\), since its zero coordinates contribute at least
\((k-1)c_t(0)=(k-1)t>1/k\); in particular \(s>0\) on \(\mathscr D_{k,t}\). Finally,
\(F(u)<\log k\): the point \((t,\ldots,t)\) has \(C=0\), so small nonconstant
perturbations of it lie in \(\mathscr D_{k,t}\), and their normalizations are not
uniform.
 
\medskip\noindent\textbf{Step 2.}
Suppose that \(u_{i_0}=0\), and choose positive coordinates \(u_i\leq u_j\), which exist
by Step~1. Recall that \(c_t(w)=t-2\sqrt{t(1-t)}\sqrt w+O_t(w)\) as \(w\downarrow0\). If
\(u_i<u_j\), move \(\varepsilon\) from \(u_i\) to \(u_j\) and set \(u_{i_0}=M\varepsilon^2\).
Then \(F\) decreases by \(s^{-1}\log(u_j/u_i)\,\varepsilon+O(\varepsilon^2|\log\varepsilon|)\),
while \(C\) changes by at most \(O(\varepsilon)-2\sqrt{t(1-t)M}\,\varepsilon\). If
\(u_i=u_j\), move \(\varepsilon\) and set \(u_{i_0}=M\varepsilon^4\). Then \(F\) decreases
by \(\varepsilon^2/(su_i)+O(\varepsilon^3)\), while \(C\) changes by at most
\(O(\varepsilon^2)-2\sqrt{t(1-t)M}\,\varepsilon^2\). In both cases, for \(M\) large and
then \(\varepsilon\) small, the new vector lies in \(\mathscr D_{k,t}\) and has a smaller
value of \(F\), a contradiction. Hence all coordinates of \(u\) are positive.
 
\medskip\noindent\textbf{Step 3.}
By Steps~1 and~2, \(u\in(0,1/4)^k\), where \(F\) and \(C\) are smooth, and
\[
 \partial_iF(u)=-\frac1s\Bigl(\log u_i-\frac1s\sum_ju_j\log u_j\Bigr).
\]
If \(C(u)<1/k\), then \(\nabla F(u)=0\), so all \(u_i\) are equal and \(F(u)=\log k\),
contradicting Step~1. Thus \(C(u)=1/k\). Since \(\nabla C(u)=(c_t'(u_i))_i\) vanishes only
at \((t,\ldots,t)\), where \(C=0\), the multiplier rule gives \(\mu\geq0\) with
\(\nabla F(u)+(\mu/s)\nabla C(u)=0\), that is,
\[
 f_\mu(u_i)=\frac1s\sum_ju_j\log u_j\quad(1\leq i\leq k),\qquad
 f_\mu(v):=\log v-\mu c_t'(v).
\]
Here \(\mu>0\), since \(\mu=0\) would again force all \(u_i\) to be equal.
 
\medskip\noindent\textbf{Step 4.}
We have
\[
 c_t''(v)=\frac{\sqrt{t(1-t)}}{2[v(1-v)]^{3/2}},\qquad
 f_\mu'(v)=\frac1v\Bigl(1-\frac{\mu\sqrt{t(1-t)}}{2\sqrt v\,(1-v)^{3/2}}\Bigr),
\]
and \(\sqrt v(1-v)^{3/2}\) is strictly increasing on \((0,1/4)\). Hence \(f_\mu'\) changes
sign at most once on \((0,1/4)\), from negative to positive, and \(f_\mu\) takes each
value at most twice there. Since the coordinates of \(u\) are not all equal, they take
exactly two values \(a<b\), with \(f_\mu'(a)<0<f_\mu'(b)\). Suppose that \(u_i=u_j=b\) for
some \(i\neq j\). The direction \(d=e_i-e_j\) satisfies \(\nabla C(u)\cdot d=0\), and the
second derivative of the Lagrangian \(F+(\mu/s)C\) along \(d\) equals
\[
 \frac2s\Bigl(\mu c_t''(b)-\frac1b\Bigr)=-\frac2s\,f_\mu'(b)<0.
\]
This contradicts the second-order necessary condition for a constrained minimum. Hence
\(b\) occurs exactly once, and \(u/s\) has the asserted form with \(x=b/s>1/k\).
 
\medskip
Finally, \(h_k\) is decreasing on \([1/k,1)\), since
\(h_k'(x)=\log\frac{1-x}{(k-1)x}<0\) for \(x>1/k\); moreover, \(x\leq L^*_{k,t}\). Hence \(\min_{\mathscr K_{k,t}}S_1=h_k(x)\geq h_k(L^*_{k,t})\). Conversely,
\(\Lambda_{k,t}\) is convex and permutation invariant. Averaging a maximizer of
\(\lambda_1\) over the permutations of the coordinates \(2,\ldots,k\) gives the point
\(\bigl(L^*_{k,t},\frac{1-L^*_{k,t}}{k-1},\ldots,\frac{1-L^*_{k,t}}{k-1}\bigr)\in\Lambda_{k,t}\),
whose entropy is \(h_k(L^*_{k,t})\). This proves \eqref{eq:finite-k-entropy-bound}.
\end{proof}
 
\subsection{Numerical certification}
 
\begin{lemma}[A scalar certificate for the largest eigenvalue]
\label{lem:finite-k-scalar-certificate}
For \(a\in\mb R\) and \(z\geq0\), put
\[
 q_t(a,z):=\frac{a-z+\sqrt{a^2-2az(1-2t)+z^2}}2,\qquad
 \mathcal U_{k,t}(L,z):=\frac zk+q_t(1-L,z)+(k-1)\,q_t(-L,z).
\]
If \(0<L<1\), \(z\geq0\), and \(\mathcal U_{k,t}(L,z)\leq0\), then
\(L^*_{k,t}\leq L\); that is, every \(\lambda\in\Lambda_{k,t}\) satisfies
\(\lambda_i\leq L\).
\end{lemma}
 
\begin{proof}
First, \(\max_{0\leq v\leq1}\{av-zc_t(v)\}=q_t(a,z)\). Indeed,
\(c_t(v)=t+(1-2t)v-2\sqrt{t(1-t)}\sqrt{v(1-v)}\). Writing \(v=(1+x)/2\) with
\(x\in[-1,1]\) and \(b=a-z(1-2t)\), we get
\[
 av-zc_t(v)=\frac b2-zt+\frac b2\,x+z\sqrt{t(1-t)}\sqrt{1-x^2}.
\]
Since \(z\sqrt{t(1-t)}\geq0\), maximizing the last two terms is maximizing a linear
functional over the upper unit semicircle, and
\[
 \max_{x\in[-1,1]}\Bigl(\frac b2\,x+z\sqrt{t(1-t)}\sqrt{1-x^2}\Bigr)
 =\sqrt{\frac{b^2}4+z^2t(1-t)},\qquad
 \frac b2-zt+\sqrt{\frac{b^2}4+z^2t(1-t)}=q_t(a,z).
\]
 
Now let \(u\in\mathscr D_{k,t}\), and put \(a_1=1-L\) and \(a_2=\cdots=a_k=-L\). Since
\(z\geq0\) and \(\sum_ic_t(u_i)\leq1/k\),
\[
 u_1-L\sum_iu_i=\sum_ia_iu_i
 \leq\frac zk+\sum_i\bigl(a_iu_i-zc_t(u_i)\bigr)
 \leq\mathcal U_{k,t}(L,z)\leq0.
\]
Dividing by \(\sum_iu_i>0\) gives \(\lambda_1\leq L\), and permutation invariance gives
\(\lambda_i\leq L\) for every \(i\).
\end{proof}
 
Let \(k=182\), \(t=27/100000\), \(L=162513/10^6\), and \(z=2077/2000\); all four are exact rationals. The accompanying Python script~\cite{WuLean2607} evaluates the quantities below in an interval:
\begin{enumerate}[label=\textup{(\roman*)}]
\item \emph{Hypotheses of Lemma~\ref{lem:finite-k-entropy-reduction}.} We have
\(1/(k(k-1))=3.0356\ldots\times10^{-5}<t<1-1/k\) and \(\bar u_{182,t}=0.0081905\ldots<1/4\).
\item \emph{Largest eigenvalue.} We have
\(-1.660\times10^{-8}<\mathcal U_{182,t}(L,z)<-1.659\times10^{-8}\), so
\(L^*_{182,t}\leq L\) by Lemma~\ref{lem:finite-k-scalar-certificate}.
\item \emph{One copy.} We have \(4.79748729461169<h_{182}(L)<4.79748729461170\).
\item \emph{Two copies.} With \(\alpha_{182,t}\) and \(\beta_{182,t}\) given by the closed
forms in Step~2 of the proof of Proposition~\ref{prop:certified-k182},
\[
 9.59449702812108<S_1(\lambda^{\mathrm{Bell}}_{182,t})<9.59449702812109.
\]
\end{enumerate}
By Lemma~\ref{lem:finite-k-entropy-reduction}, items (i)--(iii) give the one-copy bound in
Step~1 of that proof, and item (iv) gives the two-copy bound in Step~2. We get nonadditivity from those bounds directly.

\section*{Acknowledgment}
We thank H. Derksen and S. Szarek for helpful discussions. Our result uses a different approach from the concurrent work~\cite{zhen2026}. 

The authors acknowledge the use of ChatGPT 5.4 and 5.5 in preparing
the first version of this manuscript. The constructions and initial
ideas, including Haar-random projections and antisymmetric
postprocessing, were proposed by the authors. The AI tools contributed
to the following tasks:
\begin{itemize}
    \item Evaluating the rightmost spectral edge; see
    \eqref{eq:app-edge-support-duality} in
    Lemma~\ref{lem:bernoulli-edge-duality}.
    \item Computing the second moment of the normalized Choi state;
    see Proposition~\ref{prop:asymptotic-choi-purity}.
    \item Checking the asymptotic estimates for the output entropies
    in Appendix~\ref{app:entropy-asymptotics}.
    \item Identifying~\cite{Nechita2018} as a reference for the
    block-modified strong convergence theorem.
\end{itemize}

During revision, the authors used ChatGPT-6 Astra and Claude Fable 5.1
to assist with revising the manuscript and developing a Lean~4
formalization of its proofs~\cite{WuLean2607}. The formalization
verifies all mathematical results in the manuscript conditional on
the block-modified strong convergence
theorem~\cite[Thm.~5.2]{Nechita2018} and the result stated
in~\cite[Thm.~1.4]{CollinsMale2014}, both of which are treated as
external assumptions. Due to computational resource constraints,
the block-modified strong convergence theorem was checked and verified by the authors outside the Lean formalization.

\bibliographystyle{IEEEtran}
\bibliography{min}
\end{document}